\documentclass[journal,onecolumn,draftcls,12pt]{IEEEtran}
\IEEEoverridecommandlockouts

\makeatletter
\def\markboth#1#2{\def\leftmark{\@IEEEcompsoconly{\sffamily}\MakeUppercase{\protect#1}}%
\def\rightmark{\@IEEEcompsoconly{\sffamily}\MakeUppercase{\protect#2}}}
\makeatother

\usepackage{etoolbox}
\newtoggle{TRport}
\toggletrue{TRport}
\newtoggle{TRportMARK}
\togglefalse{TRportMARK}

\usepackage[utf8x]{inputenc}
\usepackage[english]{babel}
\usepackage{ucs}
\usepackage{amsmath}
\usepackage{amsfonts}
\usepackage{amssymb}
\usepackage{amsthm}
\usepackage{array}
\usepackage{verbatim}
\usepackage{multirow}
\usepackage{listings}
\usepackage{stfloats}

\usepackage{algorithm}
\usepackage{algorithmic}
\usepackage{url}
  \usepackage{enumerate}
  \usepackage{cite}
\usepackage[usenames,dvipsnames,svgnames,table]{xcolor}
\ifCLASSINFOpdf
  \usepackage[pdftex]{graphicx}
  \usepackage{subfigure}
  \usepackage{epstopdf}
   \graphicspath{{./}{./img/}{../matlab/MFBPsim/figures/}}
   \DeclareGraphicsExtensions{.eps,.ps,.png}
\else
  \usepackage[dvipdf]{graphicx}
  \usepackage{subfigure}
  \graphicspath{{./img/}}
   \DeclareGraphicsExtensions{.eps,.ps}
\fi

\newcommand{\Hb}{\mathbf{H}}

\newcommand{\C}{\mathbf{C}}

\newcommand{\Q}{\mathbf{Q}}

\newcommand{\one}{\mathbf{1}}

\newcommand{\x}{\mathbf{x}}

\newcommand{\s}{\mathbf{s}}
\newcommand{\q}{\mathbf{q}}

\newcommand{\y}{\mathbf{y}}

\newcommand{\w}{\mathbf{w}}

\newcommand{\lambdab}{\mathbf{\boldsymbol{\lambda}}}

\newcommand{\ab}{\mathbf{a}}

\newcommand{\pp}{\mathbf{p}}

\newcommand{\Ex}[2]{{\textnormal{E}_{#1}\left[#2\right]}}

\newtheorem{definition}{Definition}

\newtheorem{proposition}{Proposition}

  \usepackage{xcolor}
   \definecolor{blueH3}{rgb}{0,.5,1}
   \definecolor{blueH2}{rgb}{0,0.25,0.75}
   \definecolor{blueH1}{rgb}{0,0,0.5}
   \definecolor{grayOldText}{rgb}{.5,.5,.5}
   \definecolor{VCobalt}{HTML}{005682}
   \definecolor{TZTeal}{HTML}{008080}
   \definecolor{KYJade}{HTML}{008151}
   \definecolor{ARust}{HTML}{a10000}

\title{\iftoggle{TRport}{{\iftoggle{TRportMARK}{\color{red}}{}Technical Report on }}{}Optimal Link Scheduling in Millimeter Wave Multi-hop Networks with Space Division Multiple Access and Multiplexing}

\author{
Felipe G\'omez-Cuba$^1$, \emph{Member, IEEE}, Michele Zorzi$^2$, \emph{Fellow, IEEE}, 
\thanks{This work was presented in part at the IEEE Information Theory and Applications Workshop (ITA), La Jolla, February 2016~\cite{gomezITAoptimal}. $^1$:F. G\'omez-Cuba is with Dipartimento di Ingegneria dell'Informazione, University of Padova, Via Gradenigo 6/b, 35131 - Padova Italy, and with Department of Electrical Engineering, Stanford University, 350 Serra Mall, 94305 CA USA (e-mail: {\tt gmzcuba@stanford.edu}), $^2$Michele Zorzi is with Dipartimento di Ingegneria dell'Informazione, University of Padova, Via Gradenigo 6/b, 35131 - Padova, Italy. (Email: zorzi@dei.unipd.it). This project has received funding from the European Union's Horizon 2020 research and innovation programme under the Marie Sk\l{}odowska-Curie grant agreement No 704837.}
}
\begin{document} 
\maketitle
\iftoggle{TRport}{
    \markboth{Technical Report - SIGNET - University of Padova - 2018-06-22}{Technical Report - SIGNET - University of Padova - 2018-06-22}
  }{
    \markboth{DRAFT}{DRAFT}
  }
\begin{abstract}
In this paper we study the maximum throughput achievable with optimal scheduling in multi-hop networks with highly directive antenna arrays capable of Space Division Multiplexing (SDM) at the transmitter and Space Division Multiple Access (SDMA) at the receiver. This network model is relevant for future millimeter wave (mmWave) systems, which are expected to implement self-backhauled cellular networks with very high data rates, relying on carrier frequencies between 10-300 GHz, channels with a very large bandwidth, and a large number of antenna elements, even in mobile devices.

We adapt mmWave channel propagation, antenna array and link rate models to the classic throughput-optimality and Network Utility Maximization (NUM) scheduling framework for multi-hop networks. Directional antenna gains, transmission towards multiple destinations at once, and simultaneous reception of signals from multiple sources, are all new characteristics not featured in the existing NUM literature. We verify that the classic NUM convergence lemmas are still valid under this new set of constraints, and discuss a series of algorithms to achieve or approximate the performance of the optimal Maximum Back Pressure (MBP) solution. Finally, our analysis estimates the potential improvement in cellular network throughput capacity due to the integration of SDM/SDMA techniques and multi-hop.
\end{abstract}


\section{Introduction}
\label{sec:introduction}

Millimeter wave (mmWave) frequency bands have been proposed to satisfy the increasing spectrum needs in cellular wireless networks beyond 4G standard generations. The untapped mmWave spectrum offers a $200\times$ increase in available bands, allowing channel bandwidths on the order of GHz. Additional gains can be achieved using highly directional antenna arrays with a large number of elements packed in a small form factor thanks to the short wavelength. However, propagation loss at mmWave frequencies suffers from at least a $20$~dB penalty in free space over current cellular microwave systems, and even more in harsh propagation environments e.g.,  due to little wall penetration and the attenuation of scattered reflections. The increased pathloss can be partially compensated for by the increased beamforming gain at the antenna arrays, but still the range is not expected to exceed $200$~m and additional Access Points (AP) need to be deployed to extend coverage around corners, walls or buildings \cite{rappaport2013,RanRapEr:14}. These APs are likely to be wireless back-haul Relay Nodes (RN) due to the prohibitive cost of providing fully wired connections to such ultra-dense small cells. 

Thus, there is a need to develop wireless network architectures with very narrow directional transmission antenna gains, multi-hop operation and sufficient routing flexibility to adapt to the density and heterogeneous conditions of 5G cellular deployments.
\iftoggle{TRport}{{\iftoggle{TRportMARK}{\color{red}}{}
  Furthermore, multi-hop  mmWave relaying is particularly attractive for the evolution of cellular architectures due to the fact that some conventional cellular Base Stations (BS) already rely on mmWave frequencies for backhaul, in the form of dedicated point-to-point out-of-band links that replace the usual wired backhaul connection. Therefore the addition of mmWave operation to the access section of the cellular system (between the BS and the users), brings all links to the same region of the spectrum and naturally gives rise to a multi-hop mmWave cellular network.
}}

\iftoggle{TRport}{{\iftoggle{TRportMARK}{\color{red}}{}
  Apart from out-of-band wireless backhauling, c}}{C}onventional cellular networks have traditionally operated single-hop topologies only, with the BSs transmitting directly to all the users. More recently, the Third Generation Partnership Project (3GPP) introduced a specification for relays in the Long Term Evolution (LTE) standard for 4G \cite{hoymann2012relaying,parkvall2011evolution}. However, this specification cannot achieve great benefits due to the rigid uplink (UL) and downlink (DL) separation that LTE uses to divide time and frequency resources \cite{fgomez2014improvedrelaying}. 
\iftoggle{TRport}{{\iftoggle{TRportMARK}{\color{red}}{}

  In LTE, the medium is divided in synchronized frames where each frame is divided in time in 10 subframes and each subframe is a time-frequency grid divided using Orthogonal Frequency Division Multiple Access (OFDMA). Frames are strictly divided between uplink (UL) and downlink (DL) following one of two options: in Time Division Duplex (TDD) UL and DL transmissions can only occur in different designated subframes using all the bandwidth; in Frequency Division Duplex (FDD), UL and DL can only make use of designated OFDMA subcarriers dividing each sub-frame in two sub-bands. In either case the resources dedicated to DL and UL are rigidly separated and constant for all the cells. The caveat of this rigid UL/DL separation is that even the simplest form of multi-hop (a 2-hop RN) is very difficult to accommodate. In theoretical RN models, simultaneous transmission towards the BS and users is exploited, but the LTE frame UL/DL separation forbids this and the advantages of relaying are severely reduced \cite{fgomez2014improvedrelaying}.
}}{}
Steps to make frame structures more flexible have been taken in more recent versions of the 4G standard \cite{YoB:12,Huang:14,3GPP36828}, which allow a flexible and user-specific configuration of subframes to DL/UL actions, enabling new ways to optimize relay scheduling.

The reason for the 
\iftoggle{TRport}{{\iftoggle{TRportMARK}{\color{red}}{}
  physical separation of all UL and DL channel resources in conventional 
}}{rigid UL/DL separation of
}
LTE is that the 4G model does not feature very directive transmission strategies (only up to 8 antennas at most), and scheduling any pair of UL and DL links at the same time results in too much interference. 
\iftoggle{TRport}{{\iftoggle{TRportMARK}{\color{red}}{}Combining this with the fact that the transmit power of BSs is much higher than users', it turns out that the radio of a BS receiving UL signals is absolutely incompatible with the existence of any BS transmitting in any neighbor cell or sector, and thus all cells in the system must be synchronized to perform the same DL/UL cycle. 
}}{}
Due to the fact that mmWave allows much more directive transmissions, a majority of such interference constraints are directly suppressed\iftoggle{TRport}{{\iftoggle{TRportMARK}{\color{red}}{}\
 by the mismatched pointing of the antenna arrays of the receiver and its potential interferers, whereas the few remaining few potential interferers (those in propagation directions similar to the desired transmitter) can be avoided with convenient scheduling (i.e. instead of a problem when there is any neighbor transmitting, in mmWave problems arise only if one specific neighbor is selected to transmit in a specific direction at a specific time).
}}{.} By reducing interference, there is no need for a universal UL/DL cycle coordination in mmWave, and 5G networks will be able to obtain new spatial multiplexing gains through the concept of \textit{Dynamic Duplex} \cite{russellDynamic,juanScheduling}\iftoggle{TRport}{{\iftoggle{TRportMARK}{\color{red}}{}, that is, scheduling an optimal set of transmissions that can feature both UL and DL at the same time on different locations. This relaxation of the scheduling optimization domain calls for the design of schedulers to maximize the performance of the network.
}}{.}

Another open issue is that single-hop cellular systems achieve a spatial multiplexing gain through Multi-User MIMO (MU-MIMO) techniques that take advantage of different users' channel matrices to allow receivers to simultaneously receive from several transmitters using Space Division Multiple Access (SDMA) or transmitters to simultaneously transmit to several receivers using Space Division Multiplexing (SDM). 
\iftoggle{TRport}{{\iftoggle{TRportMARK}{\color{red}}{}
  These techniques are based on algebraic techniques over the channel matrix, such as singular value decomposition, that allow to create different virtual channels on orthogonal or partially-isolated vectorial sub-spaces of the signal. These simultaneous transmissions can take the form of spatial multiplexing, when all signals in different subspaces are desired by the receiver, or interference suppression, when some of the subspaces carry a projected signal that does not have to be decoded and the signal processing works to actively reduce its impact on the desired signals. 
}}{}
Unfortunately, the design of such MU-MIMO techniques is often studied exclusively under the assumption of a single-hop cell forming a logical star topology with the BS at the center and every user directly connected to it. Therefore the effectiveness of MU-MIMO techniques needs to be evaluated in the context of multi-hop networking. On the one hand, the extent of the benefits of spatial multiplexing in MU-MIMO is not clear in a context where spatial diversity is already provided by the existence of multiple routes to the destination. On the other hand, there are very few results about multi-hop performance or optimal scheduling with a MU-MIMO physical layer in the multi-hop literature. 

The currently available body of work on optimal scheduling is dominated by ad hoc and sensor networks, that have been traditional niches for multi-hop architectures. These networks are characterized by low-complexity physical layers, making it extremely difficult to find results in the multi-hop literature that are readily compatible with the future mmWave technology \cite{6615900}.
\iftoggle{TRport}{{\iftoggle{TRportMARK}{\color{red}}{}
  On the contrary, two assumptions about the physical layer that are commonplace in multi-hop scheduling papers are very far from what is expected in mmWave.
\begin{enumerate}
 \item \textbf{Power control for fixed rates:} In sensor networks, the trade-off between battery and rate is usually resolved favoring of the former over the latter. Nodes adjust transmit power to the minimum necessary to reach the receiver, making received power constant in all links. Classic results on scheduling analysis exploit this property by modeling networks with normalized unit rate for all links. With this assumption, calculating throughput capacity is equivalent to counting the number of active links, and results on fundamental graph theory have a direct correspondence to throughput capacity measures. However, mmWave and cellular literature emphasizes high rates, whereas BSs usually have access to the power network. Received power and link capacity vary between links, and power allocation pursues a balance between maximizing the cell throughput and offering fair rates to all users.
 \item \textbf{Destructive collision model:} Since the earliest ALOHA protocols, interference has been canonically represented in most multi-hop models as the impossibility to recover either signal when two or more transmitters are active within range of the receiver. More flexible models introduce a limited physical layer awareness, in the sense that the Signal to Interference plus Noise Ratio (SINR) is measured and the model only declares a collision if SINR is below a threshold. However, even these models still view a collision as a fundamentally destructive event and force schedulers to avoid collisions as much as possible. In contrast, modern physical layer MU-MIMO techniques can receive multiple packets simultaneously and cancel the interference from undesired signals. This requires a new way of thinking in the design of network schedulers; one in which collisions are seen as fundamentally beneficial up to the point of the receiver's processing ability, and the scheduler must actively increase the number of simultaneous transmissions to maximize parallelism and spatial multiplexing. 
\end{enumerate}
}}

In this paper we revisit the scheduling framework of Network Utility Maximization (NUM) with Maximum Back Pressure (MBP) scheduling \cite{Tuto,juanScheduling}. This body of work studies the \textit{throughput capacity region} of the network, defined as the supremum rate region achievable in a network with given constraints, where depending on the specific constraints adopted the region may or may not correspond to a proper information-theoretic capacity region. 
%

\subsection{Related Work}
\label{sec:related}

\subsubsection{mmWave and MU-MIMO}

MmWave propagation and channel characteristics are studied in \cite{rappaport2013,RanRapEr:14,Rappaport2015}. Using these measurements, simulations are implemented to estimate mmWave cellular rates in a single-hop Urban Micro-Cell network in \cite{Akdeniz2013,Akdeniz2014}. Moreover, some mmWave signal sensing strategies to detect neighbors are reported in \cite{Barati2015}. Antenna array architectures for beamforming and multiple-user reception have been developed in many proposals such as \cite{Hur2013,Rappaport2014mimo,Samsung2014}; an exhaustive survey is provided in \cite{Kutty2015}. Moreover, abundant literature for large-array MIMO that is not necessarily specialized in mmWave can be found in \cite{Gesbert2007a,hoydis2011massive,Bjornson2016}. Some hardware issues of mmWave MIMO, namely the excess power consumption of Analog to Digital Converters (ADC), are tackled in \cite{Orhan2015,Mo2016,Abbas2016}.

\subsubsection{Multi-hop Scheduling}

Multi-hop scheduling was analyzed in \cite{Tass} to achieve stability in networks with single-hop traffic flows with fixed arrival rates, and random imperfect scheduling was first introduced. A congestion control technique is introduced in \cite{Kelly1997,Kelly1998} to achieve NUM by varying the traffic arrival rates, and the problem is generalized to multi-hop traffic flows in \cite{Eryilmaz2007}. Moreover, multiple extensions consider randomized power allocation \cite{ModianoPower}, QoS or delay \cite{ZhouDelay2012}, etc. A comprehensive survey may be found in \cite{Tuto}.

Some works \cite{russellDynamic,juanScheduling} have considered optimal scheduling in mmWave before, but \cite{russellDynamic} only considered centralized scheduling for a given tree topology, leaving the optimal routing/tree formation problem open. On the other hand, NUM was considered in \cite{juanScheduling} as we do, which deals with routing implicitly, but the physical layer considered in \cite{juanScheduling} did not allow SDM or SDMA, constraining spatial multiplexing and simplifying the physical layer and the scheduling problem. Nevertheless, thanks to the simplified physical layer, \cite{juanScheduling} performed an analysis of interference with fewer assumptions about the antenna array than ours.

The rest of this paper is organized as follows. Section \ref{sec:related} described some related literature. Section \ref{sec:model} describes mmWave channel, link, network and traffic models. Section \ref{sec:optimization} describes the NUM problem and the definition of MBP optimal scheduling. Section \ref{sec:approximations} describes some algorithms that can implement or approximate MBP. Section \ref{sec:numeric} provides numerical examples for a mmWave cell model and further discusses the properties of the NUM problem and the algorithms. Section \ref{sec:conclusion} concludes the paper.

\section{System Model}
\label{sec:model}
\subsection{mmWave SDM and SDMA Links}

\begin{figure*}[ht]
 \centering
 \includegraphics[width=.7\textwidth]{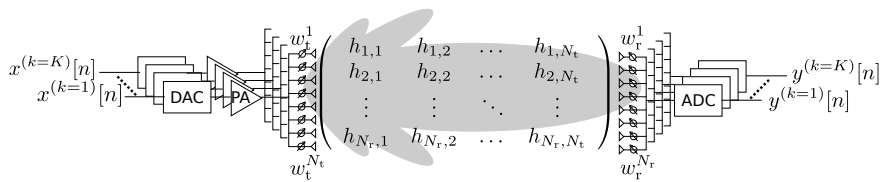}
 \caption{Analog SDM/SDMA scheme with $K$ independent transmitter-beamforming and $K$ independent receiver-beamforming signal ports. Both transmitter and receiver use each port $k\in[1,K]$ for an independent link with a different neighbor node.}
 \label{fig:analogBF}
\end{figure*}

All mmWave nodes are assumed to have arrays of $N_a$ antennas that make use of hybrid analog and digital processing techniques due to the high power consumption of ADCs, which makes full-digital MIMO difficult at high data rates \cite{Rappaport2014mimo}.
We consider the hybrid analog/digital Space Division Muliplexing/Multiple Access (SDM/SDMA) scheme represented in Fig \ref{fig:analogBF}. We assume all devices have $K\ll N$ transmission radio chains composed of DACs, power amplifiers and analog beamforming units connected to a common antenna array. Similarly, at the receive side all devices have a common antenna array connected to $K\ll N$ parallel analog beamforming units, each connected to one ADC signal port. The receiver can process signals from $K$ origins simultaneously, assigning each to one of the $K$ ports and adjusting the corresponding analog beamforming vector to the channel of the desired transmitter. The transmitter, too, can send $K$ independent signals, denoted $x_{n,1}[t]\dots x_{n,K}[t]$, to $K$ destinations at the same time, denoted by $m(n,1)\dots m(n,K)$. The transmitter $n$ must divide a total available transmit power $P_n$ among the different signals satisfying $\sum_{k=1}^{K}P_{n,m(n,k)}\leq P_n$ where $P_{n,m(n,k)}=\Ex{}{|x_{n,k}[t]|^2}$

We denote by $\mathcal{T}(m)$ the set of transmitters that have been assigned one of the $K$ receive signal ports of receiver $m$, and denote by $\mathcal{R}(n)$ the set of receivers that have been assigned one of the $K$ transmit signal ports of transmitter $n$. The port index number $k$ does not affect capacity, so for clarity we remove the index $(n,k)$ in our notation and denote the signal from transmitter $n$ to receiver $m$ as $x_{n,m}[t]$ and its power by $P_{n,m}$. In this paper we assume that $K$ is larger than the number of neighbors of all receivers $m$, i.e. all the nodes in the network whose signals can reach $m$ have been assigned one receive port. Therefore the set $\mathcal{T}(m)$ contains all the interferers for the signal received from $n$ at $m$, which we write as follows
\begin{equation}
 \begin{split}
  y_{n,m}[t]=\;&\w_{n, m}^{\mathrm{r}}\Hb_{n, m}\w_{n, m}^{\mathrm{t}}g_{n, m}x_{n,m}[t]
  +\sum_{\substack{j\in \mathcal{R}(n)\setminus m}}\w_{n, m}^{\mathrm{r}}\Hb_{n, m}\w_{n, j}^{\mathrm{t}}g_{n, m}x_{n,j}[t]\\
  &+\sum_{\substack{i\in \mathcal{T}(m)\setminus n}}\w_{n, m}^{\mathrm{r}}\Hb_{i, m}\w_{i, m}^{\mathrm{t}}g_{i, m}x_{i}[t]
  +\sum_{i\in \mathcal{T}(m)\setminus n}\sum_{\substack{j\in \mathcal{R}(i)\setminus m n}}\w_{n, m}^{\mathrm{r}}\Hb_{i, m}\w_{i, j}^{\mathrm{t}}g_{i, m}x_{i,j}[t]\\
  &+z[t],
  \end{split}
\end{equation}
where the first term is the desired transmission from $n$ to $m$, where $n$ transmits signal $x_{n,m}[t]$ with power $P_{n,m}$. The second term represents the auto-interference caused by $n$ transmitting other signals towards other receivers $\mathcal{R}(n)$, which arrive at $m$ with mismatched transmit beamforming vector, causing a residual interference. The sum transmitted power by $n$ in these first two terms adds to the power constraint $\sum_{j\in \mathcal{R}(n)}P_{n,j}\leq P_n$. The third term in the channel model represents the signals emitted towards $m$ by other transmitters $\mathcal{T}(m)$, which also leak as interference in the receive signal port of $m$ assigned to $n$ with a mismatched receive beamforming vector. The fourth term represents any other interference by transmitters that can reach $m$ ($i\in\mathcal{T}(m)\setminus n$) but are transmitting towards other destinations ($j\in\mathcal{R}(i)$) using both transmit and receive beamforming vectors mismatched on the link $n,m$. The fifth term is Additive White Gaussian Noise (AWGN) with power spectral density $N_0$. We denote by $g_{n, m}$ the macroscopic pathloss, and by $\Hb_{n, m}$ the normalized channel small scale fading matrix. Vector $\w_{n, m}^{\mathrm{r}}$ denotes the receiver beamforming at $m$'s port assigned to $n$; $\w_{n, m}^{\mathrm{t}}$ denotes the transmitter beamforming that $n$ applies to the signal transmitted towards $m$. Equivalent definitions apply to the vectors with subindex $i$ or $j$ for other transmitters and receivers in the same network.

We consider independent processing on each signal port of the receiver, and hence the receiver observes a series of $K$ scalar values $y_{n,m}$ for each transmitter $n\in\mathcal{T}(m)$. An improvement of this scheme could include the use of $K\times K$ digital MU-MIMO interference management techniques to remove terms 2 and 3. However, due to the severe shadowing, absortion and blockage in mmWave, it is expected that the channel matrices are low-rank and the beamforming gains in these mismatched interference terms is very small regardless. Some studies, such as \cite{Mudumbai2009}, propose considering mmWave links as ``pseudo-wires'',  where the effect of this interference is buried in noise and can be disregarded. Indeed, \cite{juanScheduling} compares NUM results with and without the pseudo-wired assumption, and finds the throughput capacity to be very similar in the so-called Interference-Free (IF) model and the so-called Actual-Interference (AI) model. Building on this observation, we conducted preliminary tests and found this pseudo-wired assumption holds. Embracing the pseudo-wired assumption we simplify the channel model with the approximation

\begin{equation}
\begin{split}
&\bigg|\sum_{\substack{j\in \mathcal{R}(n)\setminus m}}\w_{n, m}^{\mathrm{r}}\Hb_{n, m}\w_{n, j}^{\mathrm{t}}g_{n, m}x_{n,j}[t]\bigg|^2 +\bigg|\sum_{\substack{i\in \mathcal{T}(m)\setminus n}}\w_{n, m}^{\mathrm{r}}\Hb_{i, m}\w_{i, m}^{\mathrm{t}}g_{i, m}x_{i}[t]\bigg|^2\\
&\quad +\bigg|\sum_{\substack{j\in \mathcal{R}(n)\setminus m \\ i\in \mathcal{T}(m)\setminus n}}\w_{n, m}^{\mathrm{r}}\Hb_{i, m}\w_{i, j}^{\mathrm{t}}g_{i, m}x_{i,j}[t]\bigg|^2\ll |z[t]|^2\\
 &\quad\Rightarrow\quad y_{n,m}[t]\simeq \w_{n, m}^{\mathrm{r}}\Hb_{n, m}\w_{n, m}^{\mathrm{t}}g_{n, m}x_{n, m}[t]+z[t]
\end{split}
\end{equation}

For the calculation of the beamforming vectors, we assume that channel matrices remain constant for the duration of the scheduling algorithm, and that transmitters design the beamforming vectors to maximize the Signal to Noise Ratio (SNR) in the absence of interference.
\begin{equation}
 \w_{n, m}^{\mathrm{r}},\w_{n, m}^{\mathrm{t}}=\arg \max |\w^{\mathrm{r}}\Hb_{n, m}\w^{\mathrm{t}}|^2 \quad\textnormal{s.t.}\quad|\w_{n, m}^{\mathrm{r}}|^2=1,\quad|\w_{n, m}^{\mathrm{t}}|^2=1
\end{equation}
Since the channel is essentially static and beamforming does not depend on interference, a node $n$ can obtain the set of neighbors connected to it ($\Omega(n)$), and compute all the necessary beamforming vectors at the start of the scheduling process ($ \w_{m, n}^{\mathrm{r}},\w_{n, m}^{\mathrm{t}}, m \in \Omega(n)$). An example of mmWave neighbor detection scheme is provided in \cite{Barati2015}.

For an accurate model of mmWave propagation, we 
\iftoggle{TRport}{{\iftoggle{TRportMARK}{\color{red}}{}
  compute the macroscopic pathloss of each link with distance $d(n,m)$ in two steps. First, for each link, a state distribution is generated with three states: Outage (OUT), Line of Sight (LOS) and Non-LOS (NLOS). Second, the pathloss of the link is calculated depending on its state 
\begin{equation}
\begin{split}
 g_{n, m}(\textnormal{dB})&=\begin{cases}
		\infty & \textnormal{w.p. }p_{OUT}\\
		61.4+20\log_{10}(d(n,m))+ \log \mathcal{N}(0,5.8)& \textnormal{w.p. }p_{LOS}\\
		72.0+29.2\log_{10}(d(n,m))+ \log \mathcal{N}(0,8.7)&  \textnormal{w.p. }1-p_{OUT}-p_{LOS}
              \end{cases}\\
              &p_{OUT}=1-\min(1,e^{-0.0334d(n,m)+5.2})\\
              &p_{LOS}=(1-p_{OUT})e^{-0.0149d(n,m)}
\end{split}
\end{equation}
and the small scale fading matrix $\Hb$ is generated using the random geometric model described in \cite[Sec. III]{Akdeniz2014}
\begin{equation}
\Hb_{n,m}=\frac{1}{L}\sum_{k=1}^{N_c}\sum_{\ell=1}^{N_p}g_{k\ell}\ab_r(\theta_r^k+\theta_r^{\ell})\ab_t^{T}(\theta_t^k+\theta_t^{\ell})
\end{equation}
where $N_c\sim \textnormal{Poisson}(1.9)$ is the number of independent propagation ray clusters between nodes $n$ and $m$. Each of these clusters is composed of $N_p=20$ different and independent-amplitude scattered reflection rays. Each cluster characterizes a bundle of propagation paths that leave the transmitter with mean Angle of Departure (AoD) $\theta_r^{k}\sim U[0,2\pi)$ and arrive at the receiver with mean Angle of Arrival (AoA) $\theta_r^{k}\sim U[0,2\pi)$. Each path in the cluster has AoD and AoA slightly off of the cluster average, with root mean square angular spread $\theta_{RMS}\sim \textnormal{Exp}(10^o)$. The path angular variations are generated as wrapped Gaussians $\theta_t^{p},\theta_r^{p}\sim \textnormal{Wrapped}(\mathcal{N}(0,\theta_{RMS}))$. Finally, for each path in each cluster, the model generates an independent scalar random fading gain $g_{k\ell}\sim\mathcal{CN}(0,1)$ and a spatial signature vector for the antenna arrays that depends on the angles. For a linear array with elements separated half a wavelength, we have that both signatures are
$$\ab_\Box(\theta_\Box)=\frac{1}{\sqrt{N_\Box}}\left(0, e^{-j\pi\sin(\theta_\Box)},\dots, e^{-j\pi\sin(\theta_\Box) (N_\Box-1)}\right)^T$$
where the box $\Box$ represents that we can use this expression both for subindex $t$ and $r$.
}}{
generate $g_{n, m}$ and $\Hb_{n,m}$ using the random distributions defined in \cite[Sec. III]{Akdeniz2014}, which are omitted here due to space constraints.
}

We compute the beamforming gain in the direction of the desired link $n, m$ as $$G_{n,m}=|\w_{n, m}^{\mathrm{r}}\Hb_{n, m}\w_{n, m}^{\mathrm{t}}|^2,$$

With this channel model, the capacity of a link $n, m$ in a frame $t$ with duration $T_\mathrm{f}$ is

\begin{equation} \label{eq:rate}
c_{n , m}(t) = \alpha_1  T_\mathrm{f}W  \log \left( 1 + \alpha_2  \frac{P_{n, m}(t)  G_{n,m}   g_{n , m}}{W  N_0}  \right) \textnormal{ bits/frame}
\end{equation}

where $W$ is the system bandwidth, $P_{n,m}(t)$ is the power allocated by $n$ to transmit towards $m$, $N_0$ is the noise power spectral density, and the pathloss and beamforming gains are defined above. The two coefficients $\alpha_1,\alpha_2$ are bandwidth and power penalty factors introduced to fit  any specific practical physical layer of interest to the Shannon capacity curve, and are often obtained from empirical data. For illustration purposes, in our simulations we set these values to a $-3$ dB SNR penalty and no bandwidth penalty, i.e., $\alpha_2=0.5$ and $\alpha_1=1$.

\subsection{Network and Scheduling Model}
We represent the wireless network by the directed graph $\mathcal{G}(\mathcal{N},\mathcal{L})$, where $\mathcal{N}$ is the set of nodes (Base Stations, BS; Relay Nodes, RN; and User Equipment, UE), $\mathcal{L}$ is the set of links, and $\mathcal{F}$ is the set of traffic flows in the network, indexed by $n$, $\ell$ and $f$ respectively. We denote the cardinalities of these sets as $N$, $L$ and $F$. 

UEs can attach to as many RNs or BSs as they wish, we call the set of these two Access Points (APs), but no UE-UE connections are allowed. RNs, on the other hand, can communicate arbitrarily with any RN or BS. BSs are always connected to a wired backhaul which means that they do not need to connect wirelessly between them. Each node $n$ is aware of the set of neighbors connected to it, $\Omega(n)$, and the maximum degree of the graph is $\Omega_{\max}$. All devices have $K\geq \Omega_{\max}$ transmit and receive radio chains\footnote{This assumption improves notation clarity yet can be relaxed to accommodate networks where UEs have limited hardware as $K(n)\geq \Omega(n)$ where each node may have a different number of ports $K(n)$ and needs only to locally have more ports than neighbors. In normal scenarios $K(n)$ and $\Omega(n)$ are both consistently smaller in UEs with small arrays, and our model holds.}. However, radio stages are half-duplex in nature, and thus each device can either transmit to, or receive from, all its neighbors at once, but simultaneous transmission and reception are not possible. A potentially interesting extension that we will leave for future work is the case $1<K\leq \Omega_{\max}$, which would imply that not all neighbors can be received at the same time by some nodes. 
\iftoggle{TRport}{{\iftoggle{TRportMARK}{\color{red}}{}
  However, under such constraint, the scheduling problem described in Section \ref{sec:optimization} would not be separable and pose a much more difficult challenge. We believe a limit in the number simultaneous links could be modeled as a special type of link interference and included in a future extension of our results to a link model without the pseudo-wired assumption.
}}

For each node $n$ we define the boolean transmission role indicator $s_n(t)=1$ if node $n$ transmits at time $t$, and $0$ otherwise. Moreover, for each pair of nodes that form a link $\ell=(n,m)$, $n,m \in \mathcal{N}$, we define the normalized link power allocation $p_{n, m}(t)\in[0,1]$ to represent the fact that $n$ transmits towards $m$ with a fraction of its power equal to $p_{n, m}(t)=\frac{P_{n, m}(t)}{P_{n}}$. This means $n$ transmits to destination $m$ at time $t$ with power $P_{n,m}=p_{n, m}(t)P_{n}$. It is clear that to satisfy the half-duplex constraint $p_{n, m}(t)\leq s_n(1-s_m)$ (power must be zero if either $s_n=0$ or $s_m=1$). Moreover $\sum_m p_{n, m}(t)\leq1$ to satisfy the transmit power constraint at node $n$. 
\iftoggle{TRport}{{\iftoggle{TRportMARK}{\color{red}}{}
  Note that in an earlier version of this work \cite{gomezITAoptimal} we only considered SDMA but not SDM, and thus $p_{n, m}(t)$ was a binary indicator with values either 0 or 1; however, in this paper we generalize our results to multiple links per transmitter (SDM), allowing $p_{n, m}(t)$ to represent a real-valued fraction of power allocated by the transmitter to multiple receivers at once.
}}

We represent the state of all nodes in frame $t$ by the binary vector $\s(t)$ and we denote the power allocations for all links by the vector $\pp(t)$ with $p_{n, m}(t)$ in the $(n-1)N+m$-th index. We call the pair $(\s(t),\pp(t))$ \textit{a schedule} on the network. Note that in our terminology a \textit{schedule} $(\s(t),\pp(t))$ is the allocation for one frame $t$, and a \textit{scheduling policy} is the method that chooses all schedules $(\s(t),\pp(t))\;\forall t$. For each vector $\s(t)$, the set of all power allocations possible conditioned on $\s(t)$, $\Pi(\s(t))$, is continuous and convex and constrained by the half-duplex and power constraints. The set of all vectors $\s(t)$ is countable and contains all $2^N$ binary vectors of $N$ elements. We denote the set of all possible schedules in the network by $\Pi={\displaystyle \bigcup_{\substack{\forall \s(t)} }} \Pi(\s(t))$. Notice that given $p_{n,m}(t)=p_{m,n}(t)=0\;\forall m\in\Omega(n)$ the value of $s_n(t)$ is irrelevant to the evolution of the network, and in fact the only relevant elements of $\s(t)$ can be inferred by the nonzero elements of $\pp(t)$. We use the redundant notation $(\s(t),\pp(t))$ to conveniently highlight the separation of the problem in two sub-problems:
\begin{enumerate}[i)]
 \item the optimal power allocation over the set $\Pi(\s(t))$ with a given pre-selection of transmitter and receiver roles of the nodes $\s(t)$, and
 \item the selection of the optimal roles of each node $\s(t)$ given the ability to obtain an optimal power allocation for any given $\s(t)$.
\end{enumerate}

Next, we define the traffic features in the network. As we said above, there are $F$ flows. Each node $n$ maintains a separate queue for each flow $f$. We denote the number of packets in each queue by $q_{n}^f$. We denote by vectors $\q_n$, $\q^f$ and $\q$ the queue lengths of all flows at node $n$, the queue lengths dedicated to flow $f$ at all nodes, and all the queues of the network indexed as $n+N(f-1)$, respectively. For each flow $f\in\mathcal{F}$, we denote by $\mathcal{S}_f$ and $\mathcal{D}_f$ the sets of sources and destinations of packets (assuming there can be one or more of each\footnote{A node that is defined as a destination of a flow will withdraw from the network all packets of that flow that reach it, thus when $|\mathcal{D}_f|>1$ anycast traffic is implemented. Our model is not intended to support broadcast or multicast traffic.}). We define the number of packets produced by source $s$ for flow $f$ during the time frame $t$ as $a_s^f(t)$. When a packet of $f$ reaches a destination, it is removed from the network. We use the following definition to characterize the average packet arrival rate.
\begin{definition}
An \textit{elastic} packet arrival process associated with flow $f$ in source node $s\in \mathcal{S}_f$ is a stochastic process with a controllable time-varying mean arrival rate injected into the network $\lambda_s^f(t)=\Ex{}{a_s^f(t)}$, with a long-term mean arrival rate $x_s^f = {\displaystyle \lim_{T \to \infty} }\frac{1}{T} \sum_{t = 1}^{T} \lambda_s^f (t)$.
\end{definition}
We denote vectors $\ab(t),\lambdab(t),\x$ as the stacked packet arrival realizations, time-varying mean, and long-term average of the packet arrival processes, respectively. 

We recall that when $(\s(t),\pp(t))$ is fixed, the link capacities between any pair of nodes $c_{n , m}(t)$ are determined by \eqref{eq:rate}. In addition to transmitting at the same time to multiple neighbors, each transmitter has the ability to dedicate a fraction of the rate of each link to serve queues of different flows. Therefore, we denote by $c_{n,m}^f(t)$ the rate of link $n,m$ that node $n$ dedicates to serving the queue $q_n^f(t)$, where the assignment of rates to flows in a link must not exceed the link capacity $\sum_{f\in\mathcal{F}} c_{n,m}^f(t)\leq c_{n,m}(t)$ and the assignment over all links must not transmit more packets than there are in the queue $\sum_{m\in\Omega(n)} c_{n,m}^f(t)\leq q_n^f\;\forall f$.

Finally, as time evolves frame by frame, the evolution of each queue can be written as
\begin{equation}
 q_n^f(t+1)=\begin{cases}
 q_n^f(t)+{\displaystyle\sum_{m\in\Omega(n)}} \left[c_{m,n}^f(t)-c_{n,m}^f(t)\right]+a_s^f(t)& n\notin \mathcal{D}_f\\
             0 & n \in \mathcal{D}_f
            \end{cases}
\end{equation}

Where, if we remove the queues of flow destinations, which are always zero, we can write the previous expression summarizing the evolution of the system in matrix notation as
\begin{equation}
 \label{eq:qupdate}
 \q(t+1)=\q(t)+(\C^T(t)-\C(t))\one_{NF,1}+\ab(t)
\end{equation}
where link capacities are properly arranged in a matrix given by $\C(t):\{c_{n+N(f-1),m+N(f-1)}=c_{n,m}^{f}(t)\}$ with removed rows and columns that correspond to flow destinations that have always zero queues instead.

\section{Throughput and NUM Optimal Scheduling}
\label{sec:optimization}
\subsection{Problem Statement}
We formulate the scheduling problem as a NUM with constraints to guarantee network stability. For each flow, we define its utility function as a continuous non-decreasing function that attributes a value $\mathcal{U}^f(R^f)$ to the successful delivery of a data rate $R^f$ bits of flow $f$ to its destinations.

We say a queue is stable if it does not grow unbounded, i.e., $\lim_{t\to\infty} q_n^f(t)<\infty$ with probability $1$, and the network is stable if all queues are stable $\lim_{t\to\infty} |\q(t)|_1<\infty$ w.p.$1$. We define the stability rate region, also known as throughput capacity region, as follows
\begin{definition}
 The \textbf{throughput capacity region} $\x\in\Lambda$ is the set of long-term average rate vectors for which there exists a scheduling policy such that the network is stable.
\end{definition}
Note that $\Lambda$ defines a capacity region, because for any $\x\notin\Lambda$ the network is unstable and by definition a positive fraction of the arrival rates $\x$ will stall in the queues for an infinite time, never reaching the destination.

When the network is stable, the long-term average rates of packets leaving the network equal the long-term average rates of exogenous traffic arrivals to the network at the sources,  $R^f=\sum_{n\in \mathcal{S}(f)}x_n^f$, and the NUM problem takes the form
\begin{equation}
\label{eq:NUMproblem}
  \max_{ \x\in \Lambda} \sum_{f=1}^{F} \mathcal{U}^f \left(\sum_{n=1}^{N} x_n^f\right)
\end{equation}

\iftoggle{TRport}{{\iftoggle{TRportMARK}{\color{red}}{}
  Since $\Lambda$ is a capacity region, the NUM problem is the computation of a single point on the frontier of the capacity region of the network, such that a given function of equivalent value between rates of different users is maximized. 
}}{}
Particularly, linear utility maximizes the sum rate, whereas using a sub-linear function such as $\mathcal{U}(r)=\frac{1}{2}\log(r)$ produces throughput-fairness, and weighted functions can be employed to implement priority and Quality of Service techniques.

\subsection{Abstract Solution: Maximum Back Pressure with Congestion Control}

We say a scheduling policy is \textit{throughput optimal} if it makes the network stable for all $\x\in\Lambda$. From the definition, it follows that the solution to the NUM problem can be achieved in a network decoupling the selection of the arrival rates, $\x$, and the scheduling in the network, operated by a throughput optimal scheduler to guarantee stability independently of $\x$.

\begin{proposition} 
\label{pro:MBP}
The \textit{Maximum Back Pressure Scheduling} algorithm (Alg \ref{alg:gmbp}) is throughput optimal.
\begin{algorithm}
\small
\caption{MBP}
\label{alg:gmbp}
\begin{algorithmic}
\FORALL {$t$}
  \STATE {\begin{equation}\label{eq:MBP}\displaystyle  (\s(t),\pp(t))=\arg \max_{\substack{(\s(t),\pp(t))\\\pp(t)\in \Pi}} \sum_{n=1}^N\sum_{m=1}^N\max_f \underset{{w_{n,m}^f}}{\underbrace{c_{n,m}^f(t)(q_{n}^f-q_{m}^f)}}\end{equation}} 
  \STATE {\textbf{s.t. }}
  \STATE {$f_{n,m}^*=\arg\max_{f}(q_{n}^f-q_{m'}^f)$}
  \STATE {$\xi_{n,m}=\frac{c_{n,m}}{\displaystyle\sum_{m':f_{n,m}^*=f_{n,m'}^*} c_{n,m'}}$}
  \STATE {$c_{n,m}^f=\begin{cases}
		  \min(c_{n,m},q_{n}^f\xi_{n,m})& f=f^*_{n,m}\\
		  0& \text{otherwise}
		\end{cases}$}
\ENDFOR
\end{algorithmic}
\end{algorithm}
\end{proposition}
\begin{proof}
The proof is a variation of the proof in \cite{Eryilmaz2007}. The result for MBP is a simplified case of the proof of Proposition \ref{pro:PAC}. The details can be consulted in\iftoggle{TRport}{{\iftoggle{TRportMARK}{\color{red}}{}
  appendix \ref{app:PAC}.
}}{
 \cite{trportSDMAscheduling}.
}
\end{proof}

Throughput optimality alone is often studied on inelastic traffics with fixed arrival mean rates $\lambdab(t)=\lambdab=\x$. In such scenario, the sources must know $\Lambda$ beforehand in order to choose an $\x$ that maximizes \eqref{eq:NUMproblem}.  In order to achieve both stability and maximization of the NUM problem without a priori knowledge of $\Lambda$, it is more convenient to consider an adaptive Congestion Control solution with elastic traffic that adapts the values of $\lambdab(t)$ to the state of the queues of the sources (a subset of $\q$). This solution allows the CC to seek for the optimal $\x$ at run-time (i.e., the CC naturally evolves to the optimal $\x$ over a large number of frames). 

Let $\x^*$ be the exact solution to \eqref{eq:NUMproblem} in a network with scheduling operated by the MBP. Let us consider the following approximation of the NUM problem with a multi-objective optimization in both maximum utility and minimum queue lengths, weighted by an arbitrarily large scalar $V$
\begin{equation}
\label{eq:EPSproblem}
 \x^{V}=\arg\max_{\x} V \sum_{n,f}\mathcal{U}(x_{n}^{f})-\Ex{\q}{\q^T\x}\;
\end{equation}

\begin{proposition} 
\label{pro:NUMCC}
In a network with MBP scheduling and rates controlled by the Adaptive NUM CC algorithm (Alg \ref{alg:anumcc}), long-term rates converge to the solution of the approximate problem $\x^{V}$ and this solution is arbitrarily close to $\x^*$ as $V\to\infty$
\begin{algorithm}
\small
\caption{Adaptative NUM CC}
\label{alg:anumcc}
\begin{algorithmic}
\FORALL {$t$}
\STATE{$C_{\max}=\max_{n,m}(c_{n,m|p_{n,m}=1})$}
\STATE{\begin{equation}
\label{eq:cc}\lambda_n^f(t)=\begin{cases}
                        \max(\min(\dot{\mathcal{U}}^{-1}(\frac{q_n^f(t)}{V},C_{\max}),0) & n\in \mathcal{S}_f\\
                        0& \text{otherwise}
                       \end{cases}
\end{equation}}
\ENDFOR
\end{algorithmic}
\end{algorithm}
\end{proposition}
\begin{proof}
The proof is a variation of the proof in \cite{Eryilmaz2007}. The result for MBP is a simplified case of the proof of Proposition \ref{pro:PACNUM}. The details can be consulted in\iftoggle{TRport}{{\iftoggle{TRportMARK}{\color{red}}{}
  appendix \ref{app:PAC}.
}}{
 \cite{trportSDMAscheduling}.
}
\end{proof}

\subsection{Discussion}
Algorithm \ref{alg:gmbp} assigns a ``queue back pressure'' to each link as a function of the state of the network, where the back pressure is defined as the largest difference between queues of the same flow at the transmitter and the receiver.  By scheduling all compatible links with the MBP, the algorithm selects the links whose simultaneous transmissions lead to the steepest decrease in queue pressure aggregated across the network.
\iftoggle{TRport}{{\iftoggle{TRportMARK}{\color{red}}{}
  Roughly speaking, when a queue of a node grows it builds up pressure and the node eventually becomes a transmitter and sends out the accumulated packets to a neighbor with shorter queues. Since the destinations always have zero packets, eventually their neighbors build large queues and transfer the packets to the destination. 
}}{}
Rather than actively optimizing packet routing, the scheduler acts only to reduce queue length and stabilize the network. The average rate is guaranteed to the destinations in the long term by a simple law of conservation of traffic in a system with finite occupation.

The idea behind the congestion control is to design a multi-objective approximate optimization, introducing a penalty for queue length in the maximization objective. 
\iftoggle{TRport}{{\iftoggle{TRportMARK}{\color{red}}{}
  Alg \ref{alg:anumcc} reacts to queue lengths by reducing $\lambda_n^f(t)$ at the source nodes of each flow $f$ when the local queue grows, and increases the rate when the queue is small and the network is likely able to accept a higher rate. The utility function governs the response to queue length in a more qualitative sense, whereas the large scalar $V$ governs the strength of the reaction to changes in queue length. 
}}{}
Higher values of $V$ make the solution to the approximate problem, $\x^V$, approach the exact optimal rate distribution $\x^*$ better. However, increasing $V$ means CC tolerates longer queues, bringing the network closer to instability, making MBP require more frames to display its long-term statistic properties. 
\iftoggle{TRport}{{\iftoggle{TRportMARK}{\color{red}}{}
  In other words, the CC+MBP scheme eventually converges for any desired precision, but the more precision is demanded, the more frames it may take.
}}

There are multiple practical issues with Algorithm \ref{alg:gmbp}. First, there are no guarantees that any satisfactory rate is achieved in the short term (any given sequence of a few consecutive scheduling frames); only the result for a very large number of frames is guaranteed. This means that MBP cannot usually be used ``as is'' in the implementation of commercial devices, especially those with mobility where the topology of the network and channels change. Nonetheless, the NUM-MBP-CC framework is a remarkable tool whose value resides in enabling the study of the fundamental throughput capacity limit of a multi-hop wireless network. In this sense, the requirement of a very long number of frames to achieve the stability region $\Lambda$ calculated is not very different from the requirement of a very large codeword length in Shannon capacity analysis.

Moreover, MBP gives a definition of the links with the highest pressure, but no method is provided to select such links.
\iftoggle{TRport}{{\iftoggle{TRportMARK}{\color{red}}{}
  That is, as far as MBP is concerned, the method to obtain the solution \eqref{eq:MBP} is irrelevant, and could still be NP-hard such as an exhaustive search. 
}}{}
In the rest of this paper we address this issue by developing some algorithms that obtain sub-optimal feasible solutions to \eqref{eq:MBP} under the exact assumptions of our model, and some relaxations of our assumptions that allow the construction of exact optimal solutions to \eqref{eq:MBP}.

\begin{table}[b]
 \centering
 \caption{Algorithms considered in this paper}
 \label{tab:algorithms}
 \begin{tabular}{c|c|c|c c}
 MU-MIMO &  $\displaystyle \max_{\s}$& $\displaystyle \max_{\pp(t) s.t. \s(t)}$ &  Acronym & Distributed\\
 \hline
 \hline
None  			& \multirow{2}{*}{Exact Algorithms}&\multirow{3}{*}{N/A}	&MWM  \\\cline{1-1}\cline{4-5}
\multirow{2}{*}{SDMA only} 	& & &SFWBF  \\\cline{2-2}\cline{4-5}
				& \multirow{\iftoggle{TRport}{4}{3}}{*}{\parbox{2.5cm}{\centering Message Passing (Feasible Solution)}}&&SFWMP& \checkmark\\\cline{1-1}\cline{3-5}
\multirow{4}{*}{SDM/SDMA}	&&Waterfilling				&MFWMP & \checkmark \\\cline{3-5}
\iftoggle{TRport}{
  				&				&\multirow{4}{*}{Fixed Power}& \textcolor{\iftoggle{TRportMARK}{red}{black}}{MFWMPOP} & \checkmark \\\cline{4-5}%
				}
				&				&\iftoggle{TRport}{ &}{\multirow{3}{*}{Fixed Power}}	&MFWMPSP & \checkmark \\\cline{2-2}\cline{4-5}%
				& \multirow{2}{*}{Mixed-Integer Linear Program}& &MFWLINOP\\\cline{4-5}
				& 		 				&&MFWLINSP\\\cline{2-5}
				& Random Pick and Compare	&Waterfilling	&MFWPAC & \checkmark \\
 \end{tabular}
\end{table}

\section{Algorithms and Approximations to MBP}
\label{sec:approximations}

In this section we distinguish three strategies to implement \eqref{eq:MBP}.
\begin{enumerate}
 \item The first strategy consists in identifying scheduling algorithms that do not satisfy \eqref{eq:MBP} but still satisfy the throughput-optimality and NUM-optimality Propositions \ref{pro:MBP} and \ref{pro:NUMCC}. The quintessential application of this strategy is the random Pick and Compare (PaC) scheduling algorithm used frequently in NUM literature \cite{Tass}.
 \item The second strategy consists in implementing a heuristic based on Message Passing (MP). MP is an iterative optimization technique that is very well known in physical layer literature and often applied in signal processing techniques such as Turbo Codes. The heuristic converges to a feasible solution (which may be a local maximum) and in some problems optimality can be proven.
 \item The third strategy consists in applying additional constraints to the network and channel model until an exact optimization framework is applicable. We find that one of the issues with our model is that the optimal power allocation varies for every different value of $\s(t)$. Therefore, if we assume a non-optimal fixed power allocation independent of $\s(t)$, the NUM scheduling can be treated as a Mixed Integer Linear Program (MILP) to optimize $\s(t)$. 
\end{enumerate}

The complete list of algorithms covered in this paper is given in Table \ref{tab:algorithms}, along with some of their properties. Algorithms from previous work with no multiplexing \cite{juanScheduling} and with only SDMA \cite{gomezITAoptimal} are also listed.
The only random algorithm is PaC, whereas MP and MILP are deterministic. Both PaC and MP allow for a distributed implementation, whereas MILP is centralized. We implement optimal power allocation with the waterfilling algorithm for PaC and MP scheduling, whereas fixed power allocation is imperative for MILP. In addition, it is also possible to use MP with fixed power to obtain a heuristic with additional simplicity.

%
%
%

\subsection{Pick and Compare}

The optimality of MBP is achieved in a stochastic sense for a large number of frames, so an algorithm that selects suboptimal schedules would still be considered NUM-optimal as long as queues are stable and the long-term average rate approaches $\x^V$. This is exploited in \cite{Tass} to design the Pick-and-Compare algorithm, which employs only randomly selected schedules, a one-frame memory to store the previous schedule, and simple comparison operations. This algorithm achieves a throughput capacity that is indistinguishable from MBP in a sufficiently large number of frames, even though on each particular frame $t$ the selected schedule $(\s(t),\pp(t))$ is not necessarily the one which maximizes the queue pressure. PaC is detailed in Algorithm \ref{alg:pac}.

\begin{algorithm}
\small
\caption{PaC scheduling in instant $t$}
\label{alg:pac}
\begin{algorithmic}
 \STATE {Store previous schedule transmitter roles $\s(t-1)$}
 \STATE {Generate new random schedule transmitter roles $\tilde{\s}(t)$}
 
 \IF {$\max_{\pp(t)\in \Pi(\s(t-1))} \sum_{n=1}^N\sum_{m=1}^Nc_{n,m}(t)Q_{n,m}>\max_{\pp(t)\in \Pi(\tilde{\s}(t))} \sum_{n=1}^N\sum_{m=1}^Nc_{n,m}(t)Q_{n,m}$}
  \STATE {Repeat the previous schedule $\s(t)=\s(t-1)$}
 \ELSE
  \STATE {Use the new random schedule $\s(t)=\tilde{\s}(t)$} 
 \ENDIF
\end{algorithmic}
\end{algorithm}

The propositions for throughput-optimality and NUM-optimality can be extended to the PaC algorithm.

\begin{proposition}
\label{pro:PAC}
The \textit{Pick and Compare} algorithm (Alg \ref{alg:pac}) is throughput optimal.
\end{proposition}
\begin{proposition}\label{pro:PACNUM}
In a network with PaC and Adaptive NUM CC (Alg \ref{alg:anumcc}), long-term rates converge to the solution of the approximate problem $\x^{V}$ and this solution is arbitrarily close to $\x^*$ as $V\to\infty$
\end{proposition}

\begin{proof}
The proof is a minor variation of the proof in \cite{Eryilmaz2007}, which is inspired to the original proposal in \cite{Tass}. The details can be consulted in\iftoggle{TRport}{{\iftoggle{TRportMARK}{\color{red}}{}
appendix \ref{app:PAC}.
}}{
\cite{trportSDMAscheduling}.}
\end{proof}

\subsection{Message Passing}

Message Passing (MP), also known as Belief Propagation (BP), is a family of distributed algorithms for optimization problems that can be separated in independent factors depending on different subsets of the variables. The min-sum algorithm is used in problems of the form
$$\min_{s_1,s_2\dots s_N}\sum_{a=1}^{N_a}f_{a}(\mathcal{S}_a)$$
where $s_1,s_2\dots s_N$ are the variables, $N_a$ is the number of independent factors, $f_{a}$, and $\mathcal{S}_a\subset \{s_1,s_2\dots s_N\}$ represents the subset of variables that have an effect on factor $f_a$.

The min-sum MP algorithm is represented in a ``factor graph'' which is bipartite, where the two types of vertices are variables and factors, and the edges are the connections between them (Fig. \ref{fig:factorgraph}). The MP algorithm is iterative, and at each stage estimates the cost of each variable through a repeated exchange of messages. At each iteration, all variables send a message to their factors with an estimation of their cost. The factors then compute an aggregated message that measures the cost of each factor as a function of its component variables. Next, the variables receive the factor message and update their estimation of their own cost, and the process is repeated until convergence is achieved. Usually convergence means that the cost or value of the variables stops changing between iterations, maybe within a tolerance range. 

\begin{figure}[!ht]
 \centering
 \includegraphics[width=0.35\columnwidth]{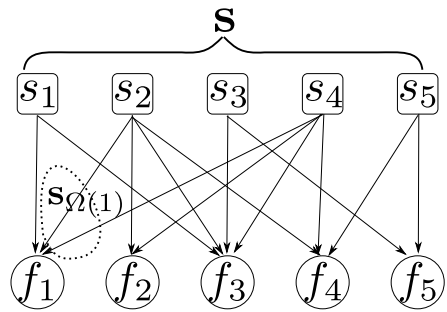}
 \caption{Message-Passing factor graph for a 5-variable 5-factor problem.}
 \label{fig:factorgraph}
\end{figure}

\iftoggle{TRport}{
  \begin{algorithm}
  \iftoggle{TRportMARK}{\color{red}}{}
  \small
  \caption{Min-Sum Message Passing}
  \label{alg:minsum}
  \begin{algorithmic}
  \STATE{$t=0$}
  \FORALL {Variable $s_i$}
    \STATE{Initial estimation of variable cost $C(s_i)^{0}$ as a function of $s_i$}
    \FORALL {Factor $f_a$ where $s_i\in\mathcal{S}_a$}
      \STATE{Initialize the message from $s_i$ to $f_a$, $M_{i,a}^{0}(s_i)$}
    \ENDFOR
  \ENDFOR
  \REPEAT
    \STATE {$t=t+1$}
    \FORALL {Factor $f_a$}
      \FORALL {Variable $s_i\in\mathcal{S}_a$}
	\STATE{Message from $a$ to $s_i$, $M_{a,i}(s_i)=\max_{\mathcal{S}_a\setminus s_i} f_a(\mathcal{S}_a)-\sum_{\mathcal{S}_i}M_{i,a}^{t}(s_i)$}
      \ENDFOR
    \ENDFOR
    \FORALL {Variable $s_i$}
      \STATE{New estimate $C^{t}(s_i)=\sum_{a=1}^{N_a}M_{a,i}^{t}(s_i)$ }
      \FORALL {Factor $f_a$ where $s_i\in\mathcal{S}_a$}
	\STATE{Message from $s_i$ to $f_a$, $M_{i,a}^{t}(s_i)=C^{t}(s_i)-M_{a,i}^{t}(s_i)$}
      \ENDFOR
    \ENDFOR
  \UNTIL{Convergence}
  \STATE{$\s^*\simeq \arg\max \sum C^{t}(s_i)$}
  \end{algorithmic}
  \end{algorithm}
}

Different variants of MP algorithms are discussed in \cite{yedidia2011message}. MP is not guaranteed to converge to a global optimum and its performance is application--dependent. There is a large collection of problems where it has been successfully applied as a heuristic, whereas there are fewer problems where its optimality, or even its convergence, is theoretically studied \cite{moallemi2010convergence}.

It must be noted that our MP implementation is one of those ``optimistic heuristic'' cases in which we do not study the theoretical characteristics of MP (although we have drawn rough inspiration from \cite{moallemi2010convergence}). Instead, we simply take the exact constraints of \eqref{eq:MBP}, build an MP-inspired heuristic with them, and observe its behavior. In our case, the MP convergence is not assured so we have added a hard maximum limit to the number of iterations and an ``oscillation detector'' that stops the algorithm when the proposed solution oscillates between two points, but the set of variables switches between two potential solutions rather than setting on one. We have observed empirically that such a two-solution oscillation occurs in our problem, and fixed the issue by storing and selecting the best of the two potential solutions when this occurs.

The major difference between our problem and the previous scheduling literature solving MBP, which limited the number of incoming links at each receiver, is that we can implement \eqref{eq:MBP} by separating the selection of node roles $\s(t)$ (transmitter or receiver) from the power allocation, that can be performed locally and independently at each transmitter, conditioned on $\s$, using the water-filling algorithm. We denote the back-pressure weight of each link and flow by $w_{n,m}^f=c_{n,m}^f(t)(q_{n}^f-q_{m}^f)$. We have that the single destination per transmitter constraint on $\pp(t)$ causes that if $c_{n,m}^f>0$ then $\forall m'\neq m,\; c_{n,m'}^f=0$. Separating the two problems, we note that for a fixed $\s(t)$ and choosing $Q_{n,m}=\max_{f}(q_{n}^f-q_{m}^f)$, then for any other flow $f'$ and power allocation $\pp'(t)$ that maximizes $w_{n,m}^{f'}$, the relation $w_{n,m}^{f'}<w_{n,m|\pp'(t)}^{f'}$ is satisfied. Thus, we can always replace $f'$ by $\arg\max_{f}(q_{n}^f-q_{m}^f)$ and increase the objective function. Conversely, if we fix $f$ and start with any power allocation $\pp'(t)$, then by definition its weight is less than or equal to the weight contributed by the optimum $\pp^*(t)$. Thus, we rewrite  \eqref{eq:MBP} as
 \begin{equation}
 \label{eq:MBPimplem}
  \max_{\s(t)}\max_{\pp(t)\in \Pi(\s(t))} \sum_{n=1}^N\sum_{m=1}^Nc_{n,m}(t)Q_{n,m}
 \end{equation}

We apply the MP framework to find a solution to the ``outer'' optimization over $\s$ in \eqref{eq:MBPimplem}, satisfying $s_n\in\{0,1\}$, $p_{n,m}\in[0,1]$, $p_{n,m}\leq s_n(1-s_m)$ and $\sum_m p_{n,m}=1 \;\forall n$, where the optimization over $p_{n,m}$ can be solved independently for each $n$ using the water-filling algorithm. This can be written as

\begin{equation}
 \max_{\s} \sum_{n} \underset{-f_n}{\underbrace{\max_{\pp(t)\in \Pi(\s(t))} \sum_{\Omega(n)} c_{n,m}(p_{n,m})\max_{f} (q_{n}^{f}-q_{m}^{f})}}
\end{equation}

Where the problem of finding $\s(t)$ is akin to a partition of the graph in two sets (transmitters and receivers) maximizing the weight on the links that connect a node in the first set to a node in the second. We represent this as an MP factor graph where for each node $n$ we consider one variable for its transmitter-receiver role, $s_n$, and one factor for the back pressure weight contributed by node $n$ when it is in the role of transmitter $-f_n(s_n,s(\Omega(n)))=\max_{m\in\{m\in \Omega(n) s_m=0\}} c_{n,m}(1)\max_{f} (q_{n}^{f}-q_{m}^{f})$.

We cannot guarantee that the MP algorithm is optimal in general. We have the following additional observations.

\begin{itemize}
 \item MP is always optimal if the factor graph contains no loops. However, due to the fact that links are bidirectional, the scheduling problem always contains loops. If $n,m$ are neighbors, there always exist two factors $f_n,f_m$ and two variables $s_n,s_m$ such as they form a loop in the factor graph.
 \item In our simulations we have observed that MP always behaved well in the sense that in some network scenarios it achieved stability, and in others the result oscillated between two feasible solutions that differed only in a small subset of the variables. Thus, we constructed a reasonable  MP heuristic by implementing an ``oscillation check'' to stop the algorithm, a memory of the previous two states of the MP solution, and a final selection of the better among the two oscillating potential solutions. 
 \item We have conjectured that this observation is related to the model in \cite{moallemi2010convergence}, but due to the fact that our factors do not satisfy the conditions for optimality, we do not believe a formal proof is a line of work that should be pursued in detail, because even if our conjecture were correct we would at best obtain a stable suboptimal algorithm without global optimality.
\end{itemize}

\iftoggle{TRport}{{\iftoggle{TRportMARK}{\color{red}}{}
  In order to compare our results in more detail with the MILP implementation described in the sequel, we have modified our MP proposition to also adopt fixed constant power allocations in all links. Different problems may be defined depending on the specific value of the fixed power constraint. We evaluate the stabilized MP scheme with two power values: an Over Powered scheme, where we consider that all links have the maximum power, ignoring the sum-power constraint at the transmitter, and a Split-Power scheme, where we consider that power allocation at each transmitter is suboptimal and equal across all links regardless of $\s(t)$.
}}

\iftoggle{TRport}{{\iftoggle{TRportMARK}{\color{red}}{}
  It must be noted that the use of constant power approximations in a MP scheme does not guarantee an upper bound of the network capacity, due to the fact that, even when stabilized with constant power, the MP distributed implementation only guarantees that a local maximum of \eqref{eq:MBP} is achieved. This means that, even when considering Over Powered links that exceed the per-transmitter power constraint, in some cases the scheduling algorithm could be stuck in local maxima and deliver rates much lower than capacity. The third type of algorithm that we introduce in the sequel is the one that provides the only verifiable upper bounds to network capacity using constant maximal power at the links, at the expense of centralizing the solution.
}}

\iftoggle{TRport}{{\iftoggle{TRportMARK}{\color{red}}{}
 In our simulations we developed the following heuristic variants of the MP MBP scheme.
}}

\iftoggle{TRport}{{\iftoggle{TRportMARK}{\color{red}}{}
  \subsubsection{Single Flow Weighted Message Passing (SFWMP)}\
  \
}}

In a previous work we introduced SDMA at the receivers without allowing transmitters to employ SDM \cite{gomezITAoptimal}. This replaces the power allocation with an easier destination selection subproblem. However, even after such modification the SFWMP algorithm is still suboptimal.

\iftoggle{TRport}{{\iftoggle{TRportMARK}{\color{red}}{}
  Thanks to the simplified receiver selection not considering SDM, we were also able to use brute-force search to obtain the exact MBP solutions in reasonable computation time. We compared the exact MBP solution obtained with brute-force and the MP heuristic, and showed that a) SFWMP does not always achieve the gains of MBP; b) Both SFWMP and MPB always achieve higher gains than Maximum Weight Matching (MWM), the baseline without SDMA the in literature; and c) The SFWMP can be considered a good heuristic algorithm that achieves an average 50\% of the SDMA gains, improving rate by a factor of 1.5.
}}

\iftoggle{TRport}{}{The MP algorithms contributed in this paper are the following:}

\subsubsection{Multiple Flow Weighted Message Passing (MFWMP)}
\

In this paper we contribute an MP algorithm for the MBP problem with multiple flows at once per transmitter using SDM with water-filling power allocation (MFWMP). In simulation we have observed that in some networks this implementation can achieve very high spatial multiplexing rate gains, outperforming SFWMP. Unfortunately, due to the fact that both are only heuristics, we have also observed some network topologies where SFWMP performs better than MFWMP. Essentially, not only is the increased flexibility of waterfilling power allocation not exploited by MP, but indeed the added complexity contributes to increasing its gap from the optimal scheduler.

\iftoggle{TRport}{{\iftoggle{TRportMARK}{\color{red}}{}
  We believe the MP algorithm with water-filling is still of interest, as the rate gain it achieves in the networks where it does well is remarkable. A potential line of future research would be the detection of network topologies well suited for this algorithm, and the use of a hybrid scheme that uses SDM only if the network topology is well conditioned, and falls back to SFWMP otherwise. 
}}

\subsubsection{Multiple Flow Weighted Message Passing with Split Power (MFWMPSP)}
\

We consider a sub-optimal constant solution to the power allocation problem consisting in reserving an equal fraction of the transmitter power for each of its links. If the neighbor is available to receive, the reserved power is assigned to that link, and otherwise the reserved power is never used. With such a sub-optimal static power allocation $p_{n,m}=\frac{1}{|\Omega(n)|}$, all the link capacities can be written as a constant of the form
$$c_{n , m} = \alpha_1  T_\mathrm{f}W  \log \left( 1 + \alpha_2  \frac{\frac{P_n}{|\Omega(n)|}  G_{n,m}   g_{n , m}}{W  N_0}  \right)$$
and with this we rewrite \eqref{eq:MBP} as
\begin{equation}
\label{eq:linMBP}
\begin{split}
  \s(t)&=\arg \max_{\s(t)\in \{0,1\}^N} \sum_{n=1}^N\sum_{m=1}^N s_n(t)(1-s_m(t))c_{n,m}\max_f(q_{n}^f-q_{m}^f)\\
	&=\arg \max_{\s(t)\in \{0,1\}^N} \s^T\C\Q\one - \s^T\C\Q\s
\end{split}
\end{equation}

Where for each pair of variables $s_n(t),s_m(t)$ the problem has two independent linear bivariate factors of the form
$$f^+(s_n(t))=\sum_{m}s_n(t)c_{n,m}\max_f(q_{n}^f-q_{m}^f)$$
$$f^-(s_n(t),s_m(t))=-s_n(t)s_m(t)c_{n,m}\max_f(q_{n}^f-q_{m}^f)$$
which gives an MP problem where each factor is the linear product of two variables and a constant cost. This MP algorithm with a bivariate factor problem always converged in our tests.

\iftoggle{TRport}{{\iftoggle{TRportMARK}{\color{red}}{}
  \subsubsection{Multiple Flow Weighted Message Passing with Over Power (MFWMPOP)}
  \
  Finally we consider the same bivariate scheme as above, but with a different constant power value in each link. We removing the sum-power constraint at each transmitter and set $p_{n,m}=1$ for all active transmitter-receiver pairs. By definition this creates a relaxed network model whose capacity region strictly contains the real network, and the NUM with this scheme upper bounds the real model. However, due to the fact that MP is a heuristic, the utility and rate achieved by the MFWMPOP algorithm may in some cases be suboptimal. This algorithm is of little interest except for the sake of completitude; the valuable insights with Over Powered fixed power allocation are produced by the MILP algorithm described in the next section.
}}

\subsection{Mixed Integer Linear Programming}

The third approach to implementing \eqref{eq:MBP} consists in modifying the network model constraints until the problem fits a MILP formulation. For this
\iftoggle{TRport}{{\iftoggle{TRportMARK}{\color{red}}{}
we consider the two fixed-power allocations described above: Split Power and Over Powered. 
}}{
we consider two different fixed-power allocations: First, we consider the Split Power (SP) allocation described above, which is a sub-optimal allocation by definition. Second, we consider an Over Powered (OP) model as in MP removing the sum-power constraint and setting $p_{n,m}=1$ for all active transmitter-receiver pairs, allowing every link to use the full power of its transmitter even when more than one simultaneous transmission is performed. By definition, when interference is negligible this creates a relaxed network model whose capacity region strictly contains that of the real network, and the optimal NUM with this scheme upper bounds the real model.
}
Under these modified constraints, the MILP enables the calculation of the exact optimal MBP scheduling solutions. This is at the expense of limiting our ability to study the transmitter power model, as well as removing our ability to implement a distributed algorithm. By considering OP links, we are also able to derive a theoretical upper bound to the network throughput capacity. The case of SP is also interesting, as it has a better scheduler than MP and allows to compute the supremum network utility achievable with any SP strategy.

\subsubsection{Multiple Flow Weighted LINear program with Over Power (MFWLINOP)}
\

The scheduling is more accurately represented in this algorithm directly by a binary power allocation per link $p_{n,m}\in\{0,1\}$, where we consider that there are no sum-transmit-power constraints at the transmitters, so all links may be active at once with full power. In addition, we modify the representation of the half-duplex constraint for this scenario. If necessary, $\s(t)$ can be inferred from the values of $\pp(t)$. Note that the representation introduced here is valid only because $p_{n,m}\in\{0,1\}$ takes integer values, and is not equivalent under the general power allocation model. We consider that each link, if active, achieves a constant capacity $C_{n,m}=c_{n,m}(p_{n,m}=1)$. Therefore, the achieved capacity may be directly represented as a product between the binary indicator and a constant as follows

\begin{equation}
\begin{split}
 \max_{\pp}& \sum_{n} \sum_{m\in \Omega(n)} p_{n,m}(t)C_{n,m}Q_{n,m}\\
 &s.t.\; p_{n,m}+\frac{1}{|\Omega(n)|}\sum_{m'\in\Omega(n)}p_{m',n}\leq 1 \forall n,m\\
 \end{split}
\end{equation}

The expression $\frac{1}{|\Omega(n)|}\sum_{m'\in\Omega(n)}p_{m',n}$ equals $0$ only if all neighbors of $n$ are not transmitting to $n$, and $1$ if all transmit. And since $p_{n,m}$ takes binary values we can either have $p_{n,m}=1$ or $\frac{1}{|\Omega(n)|}\sum_{m'\in\Omega(n)}p_{m',n}>0$, but not both. Thus the constraint serves as a linear substitute for the half-duplex rule ``node $n$ cannot transmit if any of its neighbors is transmitting to it''.

The MILP toolbox of MATLAB can be used to solve the problem above in a reasonable time, although these algorithms do not allow for a distributed implementation. Likewise, a practical deployment of this model would still be feasible in centralized-control networks such as LTE.

\subsubsection{Multiple Flow Weighted LINnear program with Split Power (MFWLINSP)}
\

In our notation we have established that $p_{n,m}=\frac{P_{n,m}}{P}$ is the normalized power allocated by transmitter $n$ to receiver $m$. In the Over Power relaxation, we use a binary indicator $p_{n,m}=1$ to represent that the transmitter allocates all its power to the receiver, i.e., $P_{n,m}=P$. In the Split Power scheme, the variable $p_{n,m}\in\{0,1/|\Omega(n)|\}$ is no longer a pure binary expression due to the sum-power constraint at the transmitter. In order to write the problem as a multiplicative binary linear expression, we require a new normalized variable, defined as $\overline{p}_{n,m}=|\Omega(b)|p_{n,m}$. Replacing the normalized variable, we can write the following modification to the OP model
\begin{equation}
\begin{split}
 \max_{\overline{\pp}}& \sum_{n} \sum_{m\in \Omega(n)} \overline{p}_{n,m}(t)c_{n,m}(p_{n,m}=1/|\Omega(n)|)Q_{n,m}\\
 &s.t.\; \overline{p}_{n,m}+\frac{1}{|\Omega(n)|}\sum_{m'\in\Omega(n)}\overline{p}_{m',n}\leq 1 \forall n,m\\
 \end{split}
\end{equation}
where the normalized binary indicator of power $\overline{p}$ is used in the half-duplex constraint, but not to calculate the constant capacity of the link $c_{n,m}(p_{n,m}=1/|\Omega(n)|)$. Anything else is resolved in the same manner.

Differently from the LINOP model, which upper bounds the capacity of the network, in the LINSP case what we obtain is a lower bound to the capacity of the network (because the scheduling is optimal but the power allocation is not) and simultaneously we obtain a supremum of the family of schedulers with SP allocation (thus outperforming MFWMPSP). 

\section{Numerical Analisis}
\label{sec:numeric}

We simulate a randomly generated picocell network as in the example in Fig. \ref{fig:topo}, with 10 UEs randomly distributed in a disk of radius $200$ m, and a BS at the center. Moreover, another four wireless RNs are placed at fixed locations at $115$ m from the BS with a $90^o$ rotation. We define the minimum connectivity requirement as a maximum omnidirectional (i.e., without beamforming) pathloss of $200$ dB. This threshold, inspired by \cite{Barati2015}, is selected for a minimal rate of $10$ Mbps when the BS transmits towards a UE, both transmitter and receiver beamforming gains are $30$dB, and the radio hardware parameters are those in Table~\ref{tab:param_tab}.

Finally, we assume two traffic flows for each UE: one uplink with source at the UE and destination at the BS, and one with source at the BS and destination at the UE. All exogenous arrivals apply the congestion control algorithm specified in \eqref{eq:cc}. To select the value of the congestion control tuning $V$ we set $V=10 C_{\max}^2$ where $C_{\max}$ is the CC maximum rate in Alg. \ref{alg:anumcc} as per the discussion in\iftoggle{TRport}{{\iftoggle{TRportMARK}{\color{red}}{}
  Appendix \ref{app:PAC}.
}}{
  \cite{trportSDMAscheduling}.
}

In our previous article \cite{gomezITAoptimal} we developed a side-by-side comparison of three algorithms: we considered the reference Maximum Weight Matching (MWM) from the literature \cite{juanScheduling} in comparison with SDMA scheduling implemented with MP and by Brute Force (SFWMP and SFWBF). In the current paper we have \iftoggle{TRport}{{\iftoggle{TRportMARK}{\color{red}}{}nine}}{eight} different algorithms and for the sake of clarity we cannot include all of them in every single figure. We will perform separate comparisons focused on different characteristics of the scheduling problem. For each characteristic, we will select the algorithms that highlight the main differences.

\begin{figure}
  \centering
    \includegraphics[width=.4\columnwidth]{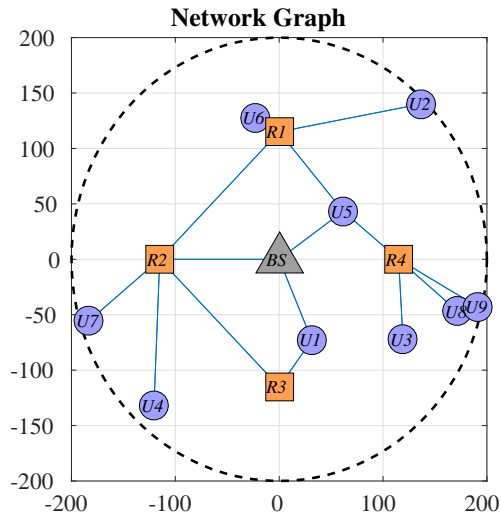}
  \caption{Picocell topology formed by all links with pathloss $\leq 200$ dB.}
  \label{fig:topo}
\end{figure}

\begin{table}[!t]
  \centering
   \caption{mmWave Channel Parameters}
   \label{tab:param_tab}
\begin{tabular}{r|l}
\textbf{Parameter} & \textbf{Values}                                          \\ \hline
Carrier Frequency  & 28 GHz                                                        \\
System Bandwidth   & 1 GHz                                                         \\
Transmission Power & 30 dBm (BS), 25 dBm (RN), 20 dBm (UE)                         \\
Noise Figure       & 5 dB (BS), 6 dB (RN), 7 dB (UE)                               \\
Antenna            & 8x8 (BS), 6x6 (RN), 4x4 (UE) $\lambda/2$  planar array \\
Connectivity       & Pathloss $< 200$ dB                                           \\
\end{tabular}
\end{table}

\subsection{Effect of Interference }

We begin by performing a sanity check on our assumption that interference is negligible as seen in Fig. \ref{fig:interfsanitycheck}. We reproduce the observations in \cite{juanScheduling} comparing the performance of mmWave networks under an interference free assumption and with an actual interference model. We conduct the test with the randomly generated network depicted in Fig. \ref{fig:topo} with a PaC scheduler for single flow without SDM/SDMA, which is the only scheduler that is properly defined in the actual-interference scenario \cite{juanScheduling}.

\begin{figure*}
  \centering
  \subfigure[PaC Single Flow, Interference Free]{
    \includegraphics[width=0.35\columnwidth]{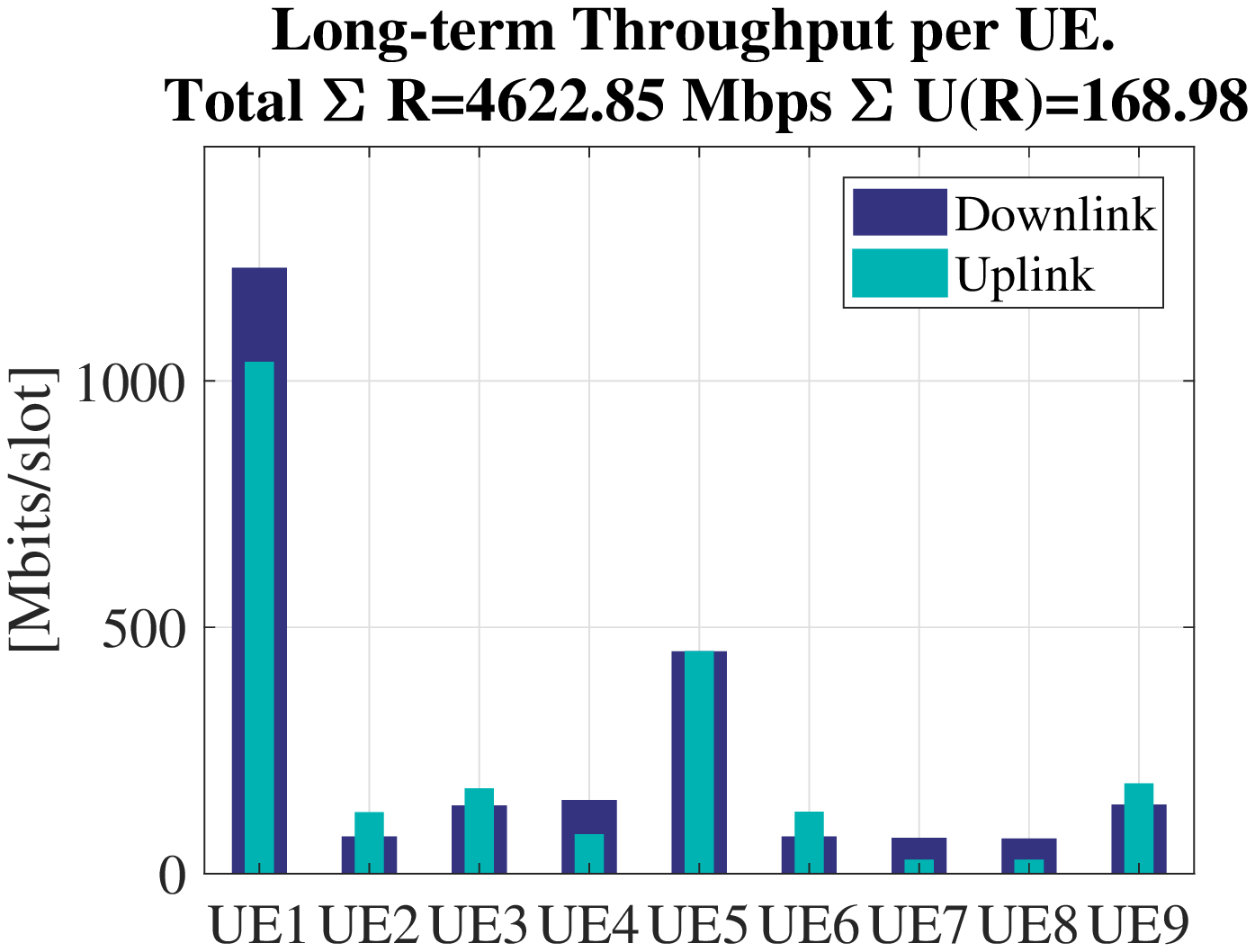}
  \label{fig:nointerf}
  }
  \subfigure[PaC Single Flow, Actual Interference]{
    \includegraphics[width=0.35\columnwidth]{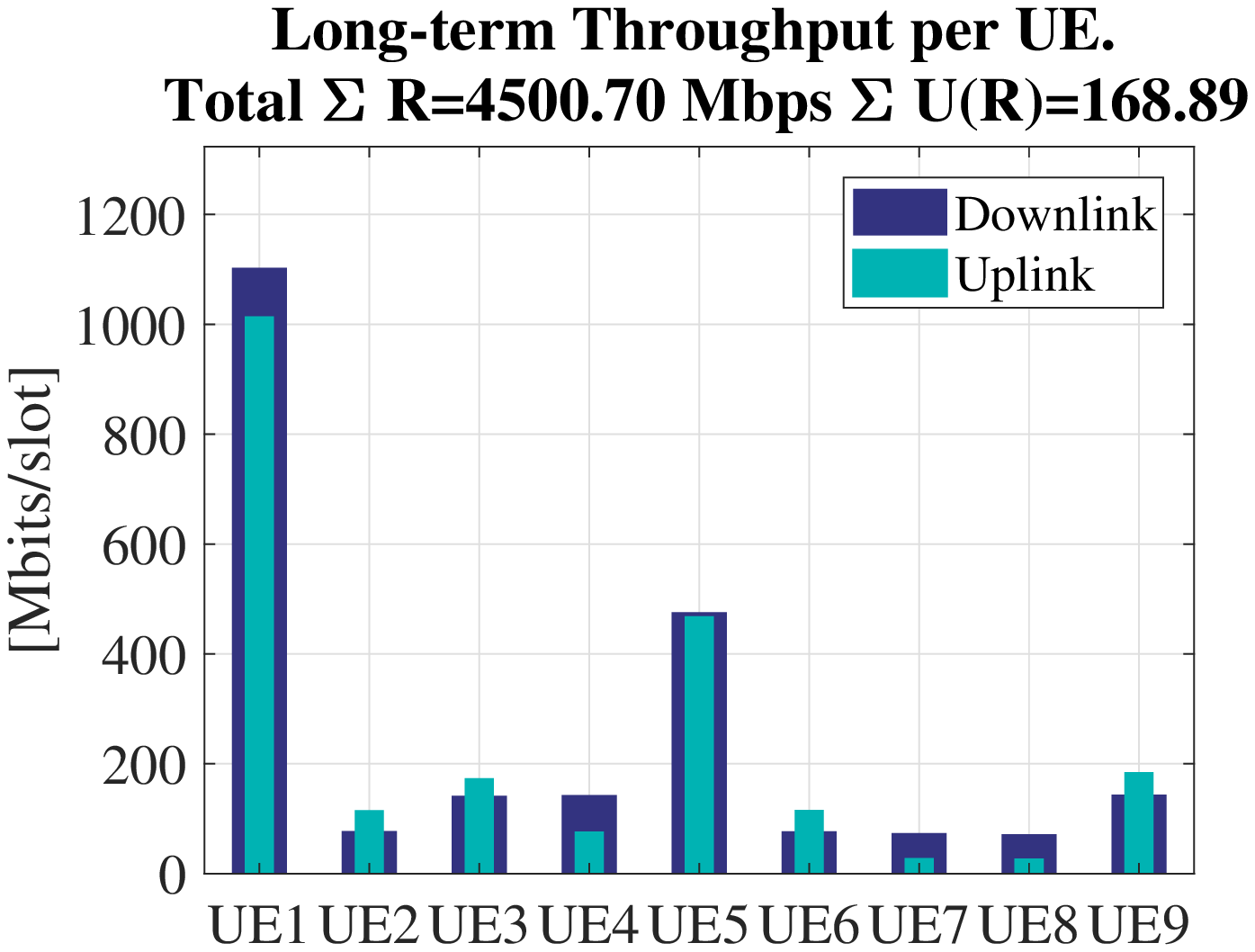}
  \label{fig:actualinterf}
  }
  \caption{Achieved user rates for the network in Fig. \ref{fig:topo} over $10^5$ frames with PaC scheduling for a single flow scheduler (no SDM/SDMA) assuming interference is negligible (Interference Free) vs same simulation with Actual Interference model as in \cite{juanScheduling}. The difference in rate results accounting for interference is minimal.}
  \label{fig:interfsanitycheck}
\end{figure*}
\subsection{MU-MIMO capabilities}

In traditional NUM literature, a one-to-one association is considered, where each transmitter can only select one destination and each receiver can only receive from one source at a time. With this constraint, all schedules are a ``matching'' of edges in the graph, and the optimal schedule can be obtained using the MWM algorithm with complexity $O(N^3)$, as in \cite{juanScheduling}. In contrast, in our previous article \cite{gomezITAoptimal} we proposed a scheme that enabled SDMA but not SDM, allowing a receiver to decode signals from multiple transmitters at once, but not the opposite. In the present paper, unlike in \cite{gomezITAoptimal}, we also allow SDM, and a transmitter can select multiple receivers at once. In Fig. \ref{fig:spatialmultiplexing} we illustrate one example schedule under each type of constraint. It must be noted that the schedules allowed under tighter constraints are a subset of the more relaxed spatial multiplexing constraints. 

\begin{figure*}
  \centering
  \subfigure[MWM]{
    \includegraphics[width=0.3\columnwidth]{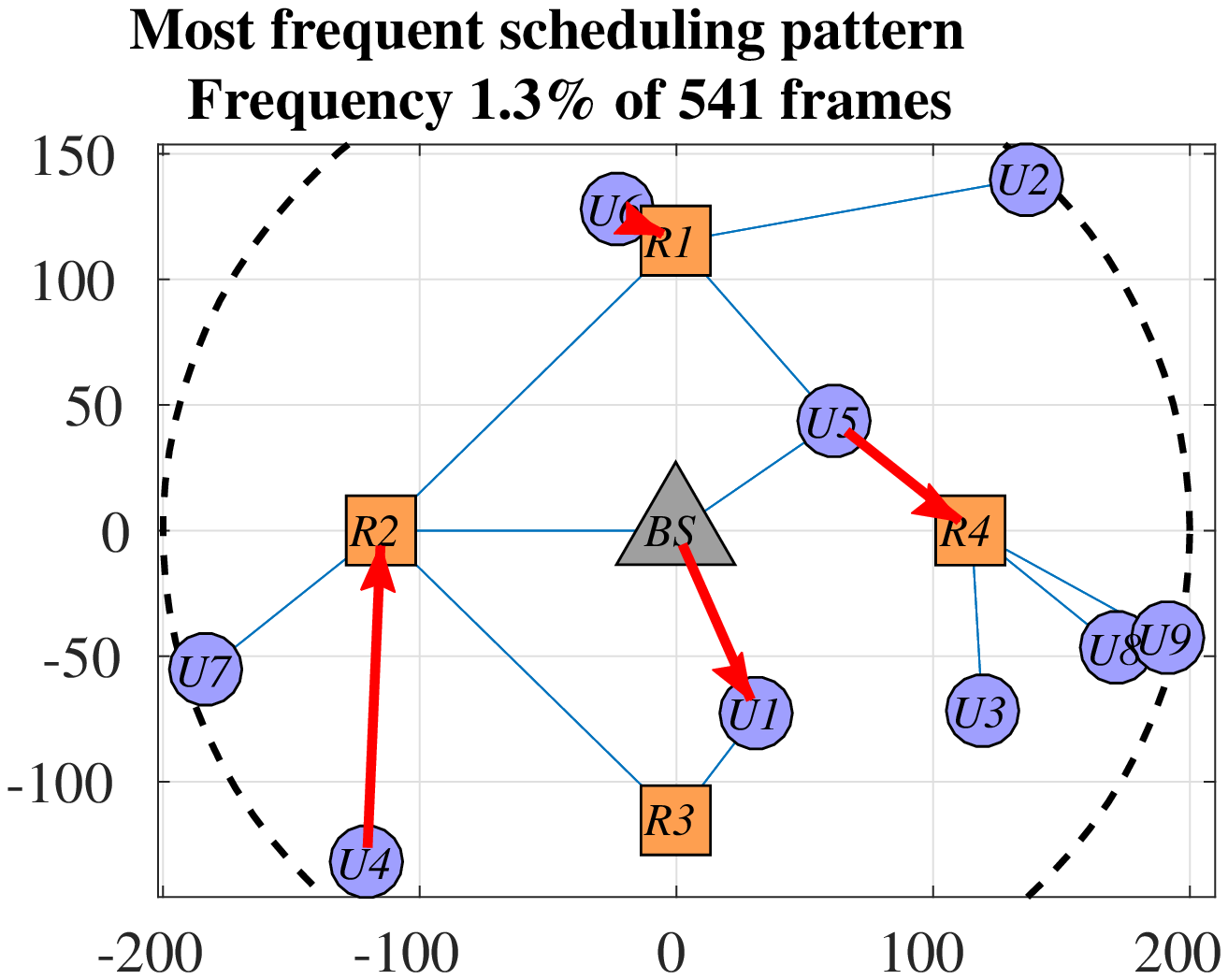}
  \label{fig:schedRef}
  }
  \subfigure[SFWMP]{
    \includegraphics[width=0.3\columnwidth]{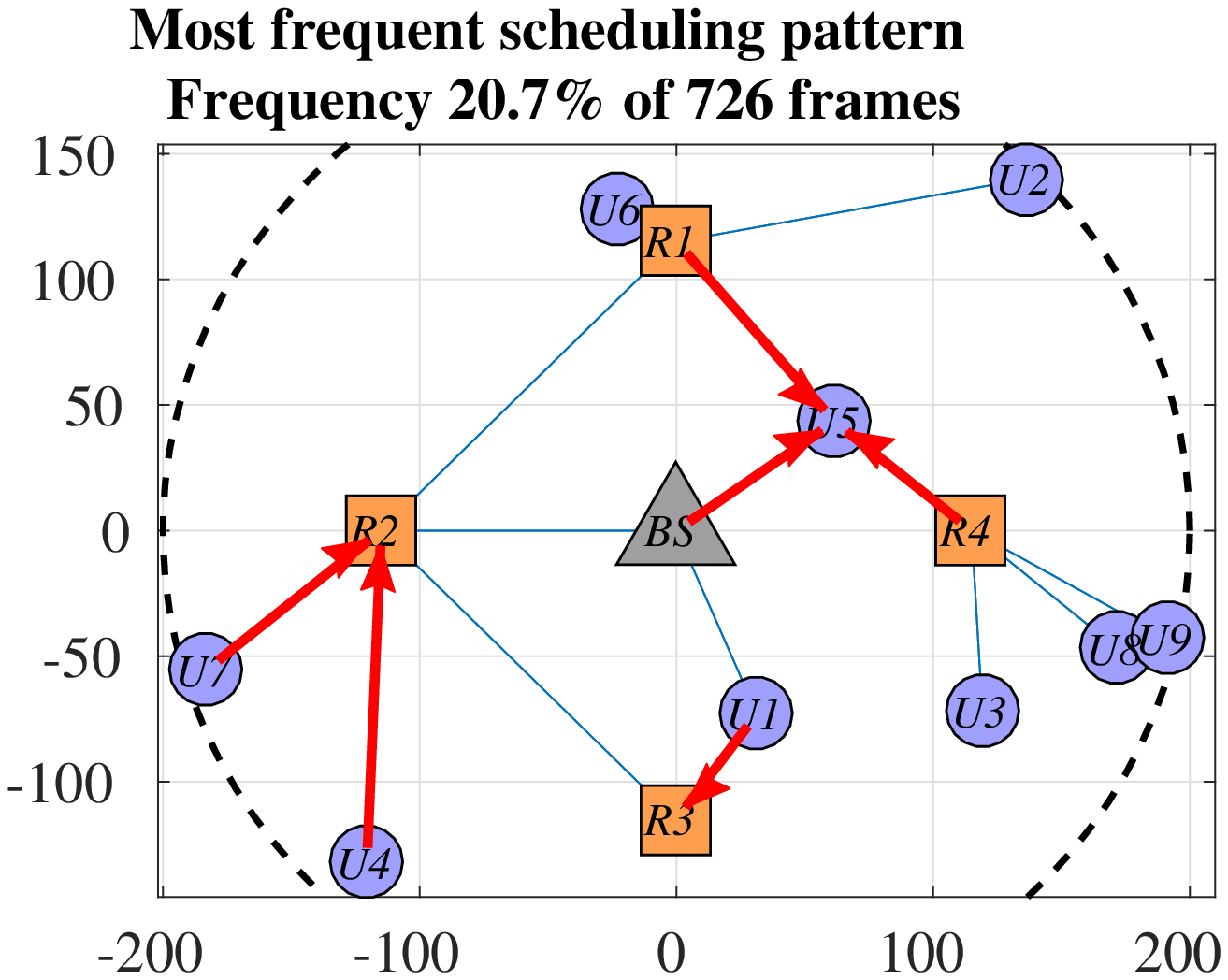}
  \label{fig:schedMP}
  }
  \subfigure[MFWMP]{
    \includegraphics[width=0.3\columnwidth]{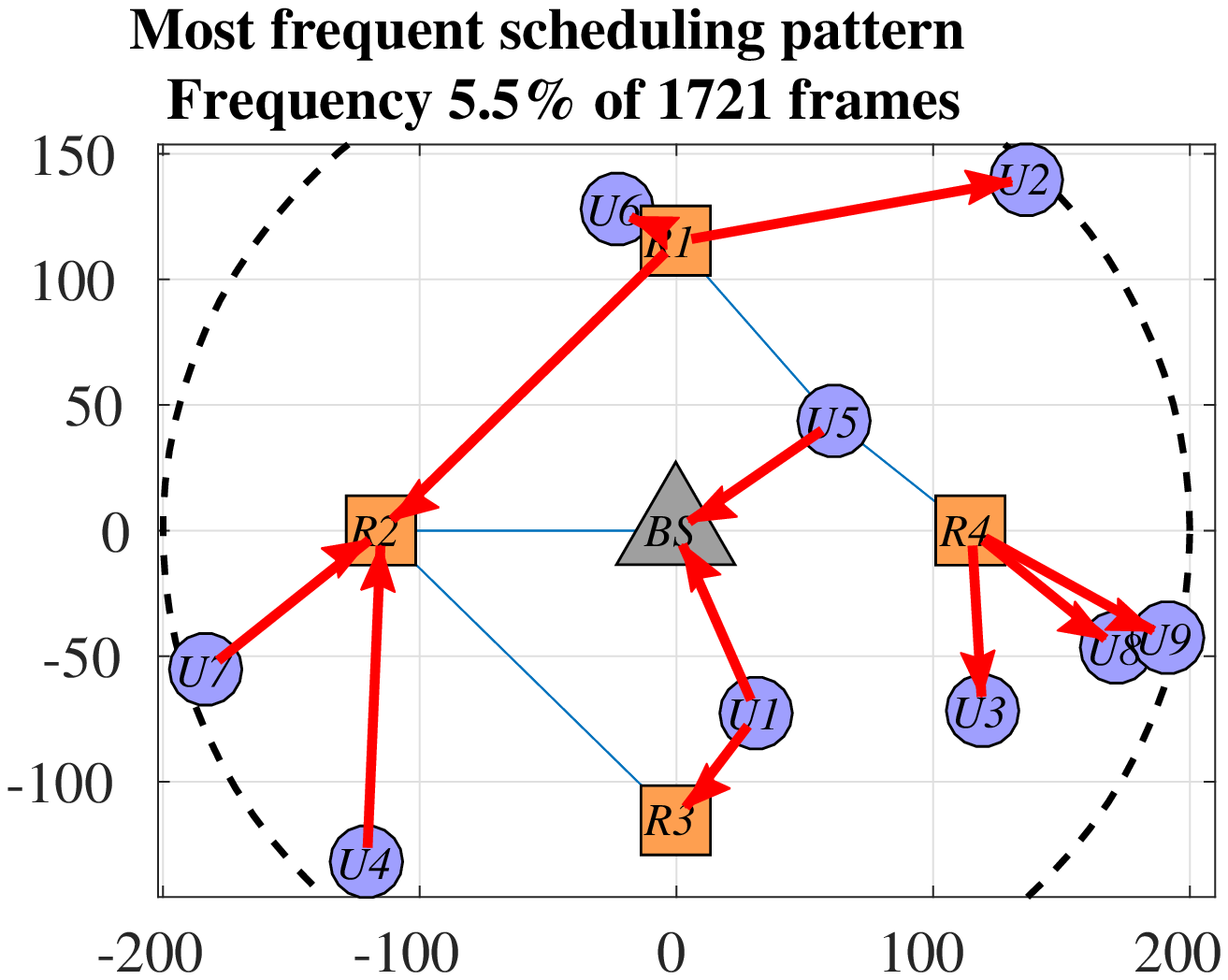}
  \label{fig:schedBF}
  }
  \caption{Example link schedules with different MU-MIMO capabilities.}
  \label{fig:spatialmultiplexing}
\end{figure*}

\iftoggle{TRport}{{\iftoggle{TRportMARK}{\color{red}}{}
  \subsection{Power Allocation}

  Considering only the MP algorithm, it can use different techniques to allocate power from each transmitter to its receivers. We represent in Fig. \ref{fig:powMP} the Network Utility and sum-rates achieved with Split Power, locally-optimal waterfilling, and Over Powered links for 10 random networks. In fig. \ref{fig:powMP} we observe that depending on the drop and the selected power allocation the MP heuristic may underperform: The ideal case is represented in drops like 3, where MP works well with all power allocation and we observe a clear progression consistent with the power model. On the other hand, in some cases like drop 4, the MP algorithm works better with some power allocations than others, and the consistent order of power allocation techniques is reversed.

  \begin{figure}
    \centering
    \subfigure[Utility]{
      \includegraphics[width=0.3\columnwidth]{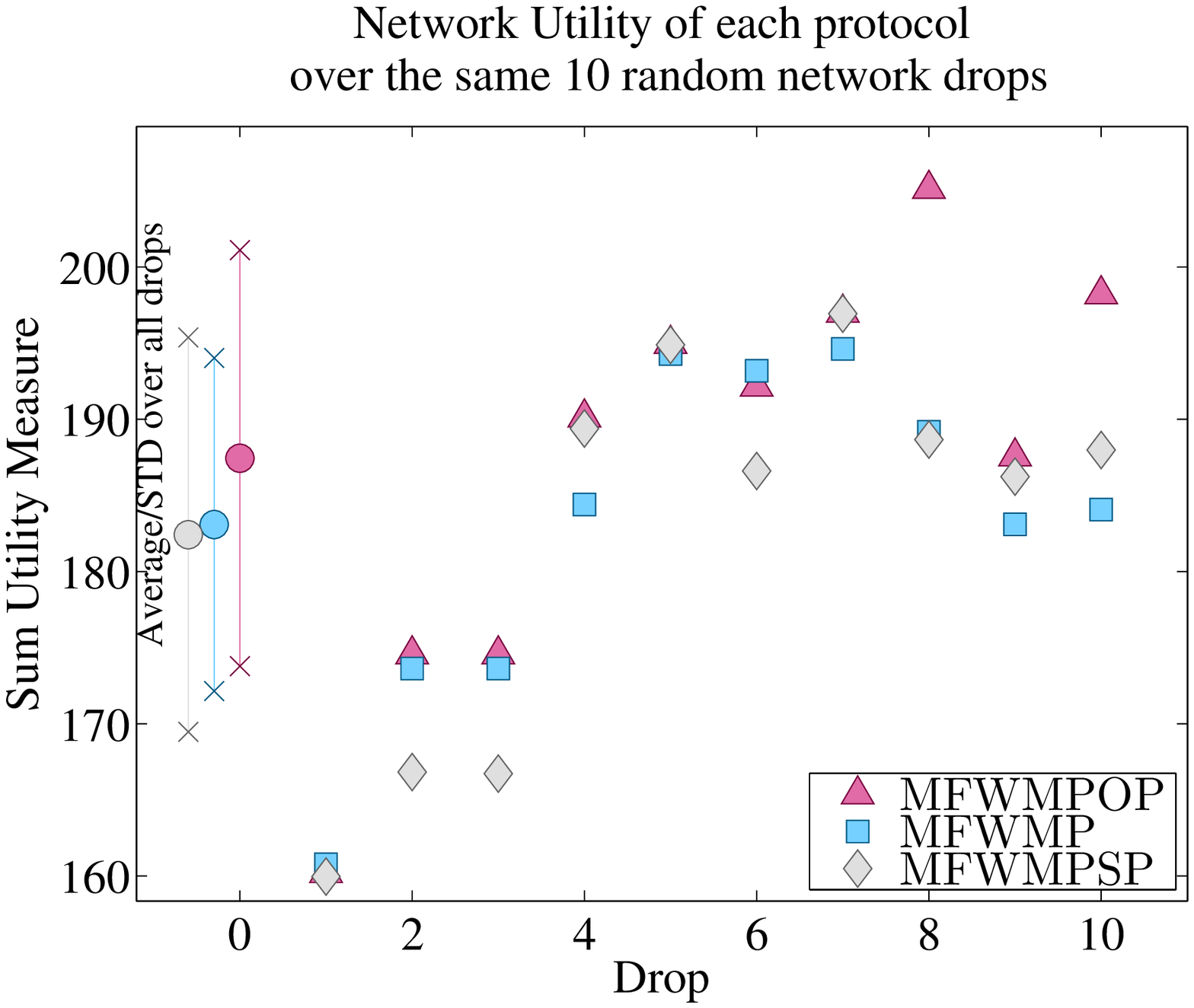}
    \label{fig:powMPutil}
    }
    \subfigure[Sum rate]{
      \includegraphics[width=0.3\columnwidth]{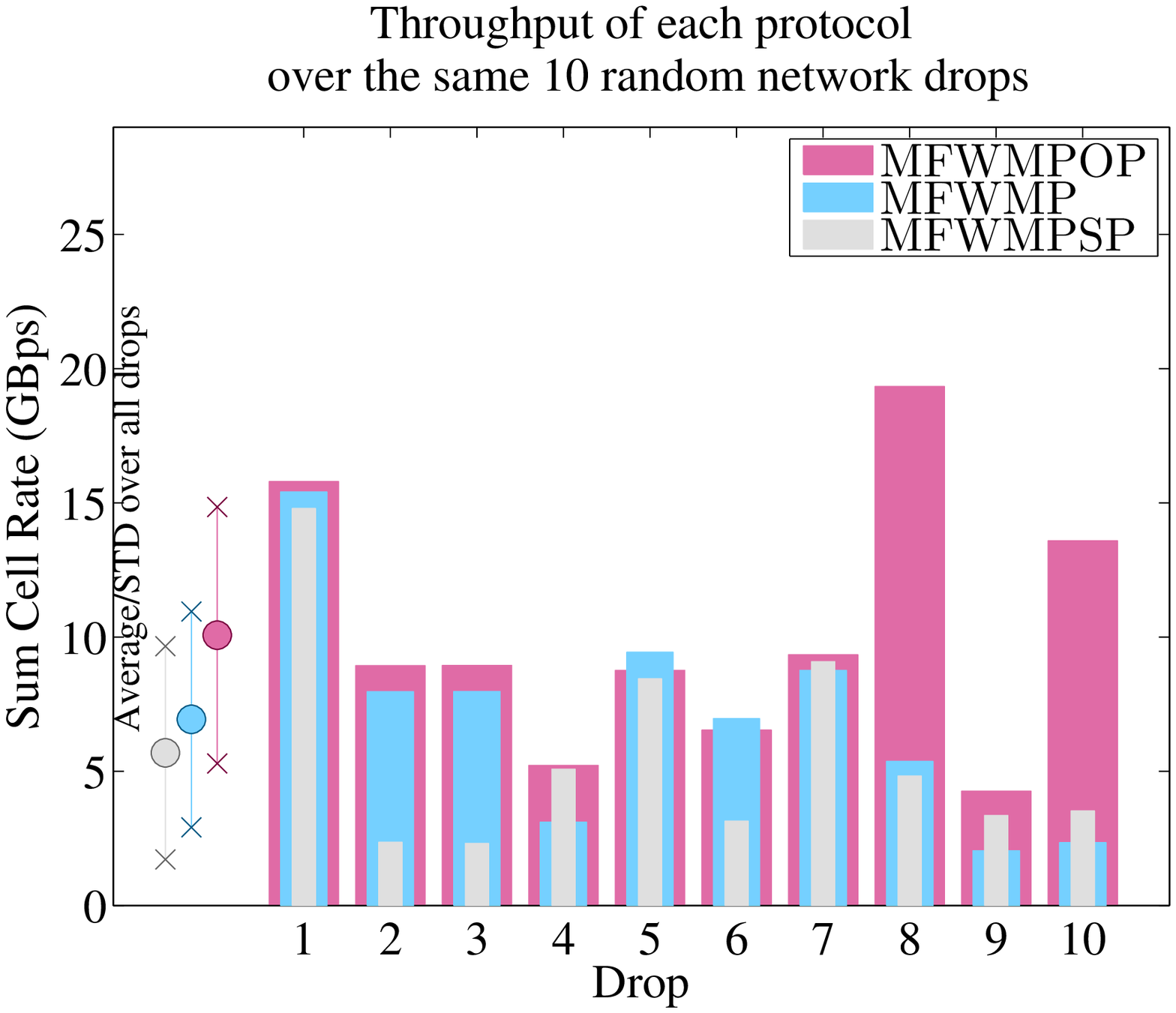}
    \label{fig:powMPrate}
    }
    \subfigure[Queues of drop 4]{
      \includegraphics[width=0.3\columnwidth]{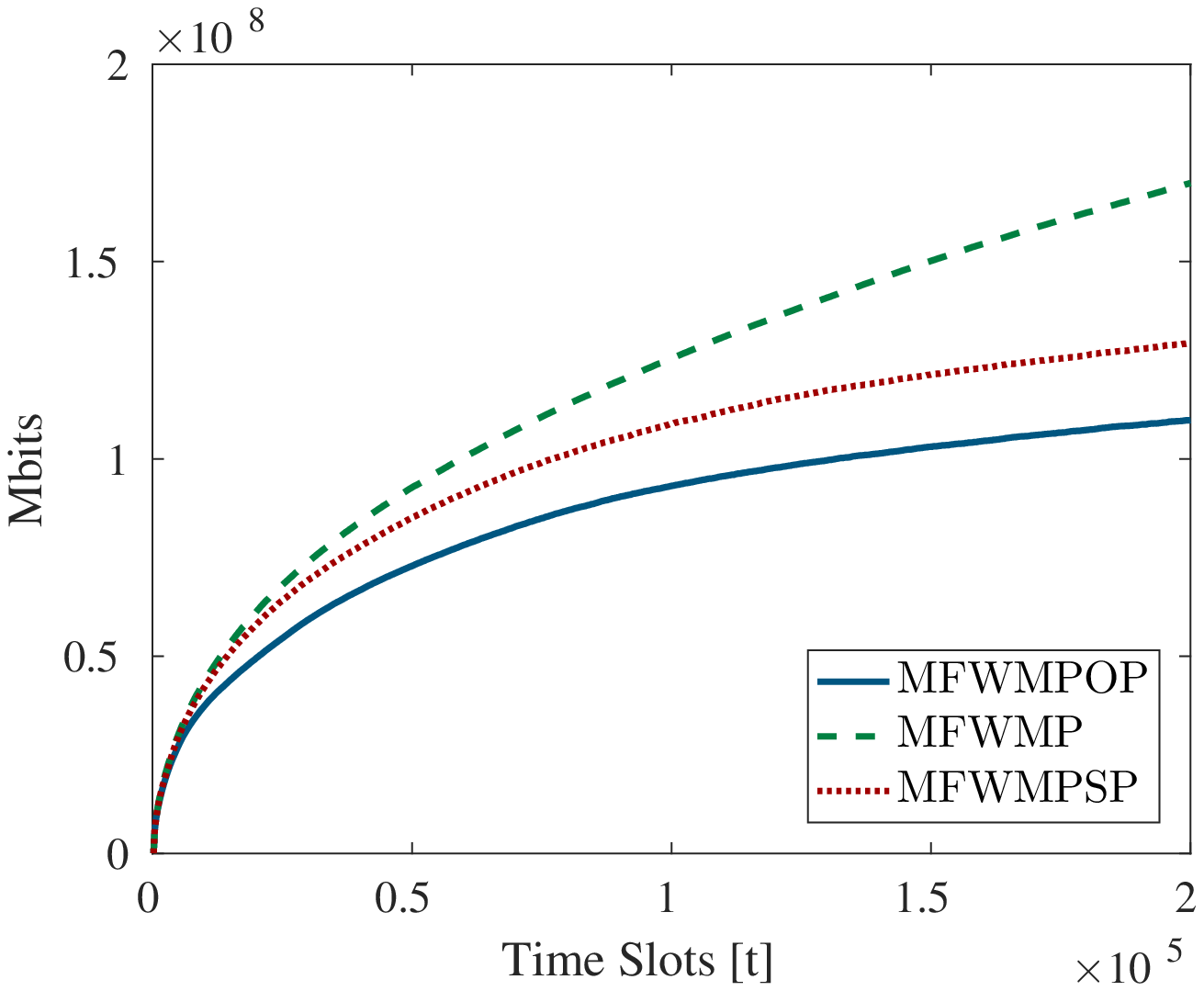}
    \label{fig:powMPqueues}
    }
    \caption{\textcolor{\iftoggle{TRportMARK}{red}{black}}{Performance of MP with three power allocation schemes.}}
    \label{fig:powMP}
  \end{figure}
}}

\subsection{Scheduling Variation}

In all simulations, the queues in the system are stabilized after a large number of frames and the scheduler experiences a steady state distribution. The most frequent states in this distribution could be replicated using practical protocols to approximate the performance of the optimal scheduler. We analyze the histogram of the different algorithms to draw some insights about how such protocols could work.
\iftoggle{TRport}{{\iftoggle{TRportMARK}{\color{red}}{}
  This could provide insights for the design of practical MAC and routing protocols that operate deterministically on much shorter frame durations, in systems constrained by mobility where the network topology varies after only a few frames and long-term effects cannot be exploited. The measure of the ``scheduling-length,'' defined as the time it takes a scheduler to serve all flows once, allows also to characterize the delay of communications in the system.
}}

In Fig. \ref{fig:histograms} we represent the histogram of each schedule $(\s_i,\pp_i)$ under different algorithms, ordered in decreasing order of number of occurrences in 200000 frames. There are two different remarkable trends to be observed here:

\begin{enumerate}
 \item First, we observe the impact of spatial multiplexing. Going from the MWM algorithm to SFWBF and MFWLINSP, all algorithms are deterministic and optimal for a given set of MU-MIMO allowed techniques. As the use of MU-MIMO increases, we see a smaller set of different schedules can cover 95\% of the behavior of the optimal scheduler. 
\iftoggle{TRport}{{\iftoggle{TRportMARK}{\color{red}}{}
 In MWM, 265 different schedules cover 95\% of the realizations of the algorithm, meaning that we could potentially reproduce 95\% of the behavior of the optimal algorithm with a deterministic MAC protocol that cycles through a series of 265 known schedules; for SFWMP, a set of only 92 different schedules can offer the same 95\%, and thus a simple MAC protocol that imitates the optimal distribution would require a 2.8 times shorter cycle. Finally, the 95\% of the time, SFWLINSP operation can be represented with a mere 12 known schedules MAC; 8 times less than SFWMP and 22 times less than MWM.
}}{}
 \item Secondly, we observe the impact of suboptimal implementations with the same level of MU-MIMO. The centralized scheme MFWLINSP spans fewer different frames than 
 the decentralized and deterministic scheme with the same fixed-power constraint, MFWMPSP\iftoggle{TRport}{{\iftoggle{TRportMARK}{\color{red}}{}, which requires 137 schedules to represent 95\% of its behavior}}.
 In turn, increasing the power allocation complexity with waterfilling (MFWMP) increases the number of frames\iftoggle{TRport}{{\iftoggle{TRportMARK}{\color{red}}{}\ up to 518}}{}. And the random scheduler MFWPAC offers the widest variation in different frames\iftoggle{TRport}{{\iftoggle{TRportMARK}{\color{red}}{}
 , with 1631 different schedules in its 95\% most frequent operations.
}}{.}

\end{enumerate}

\begin{figure*}
  \centering
  \subfigure[MWM]{
    \includegraphics[width=0.3\columnwidth]{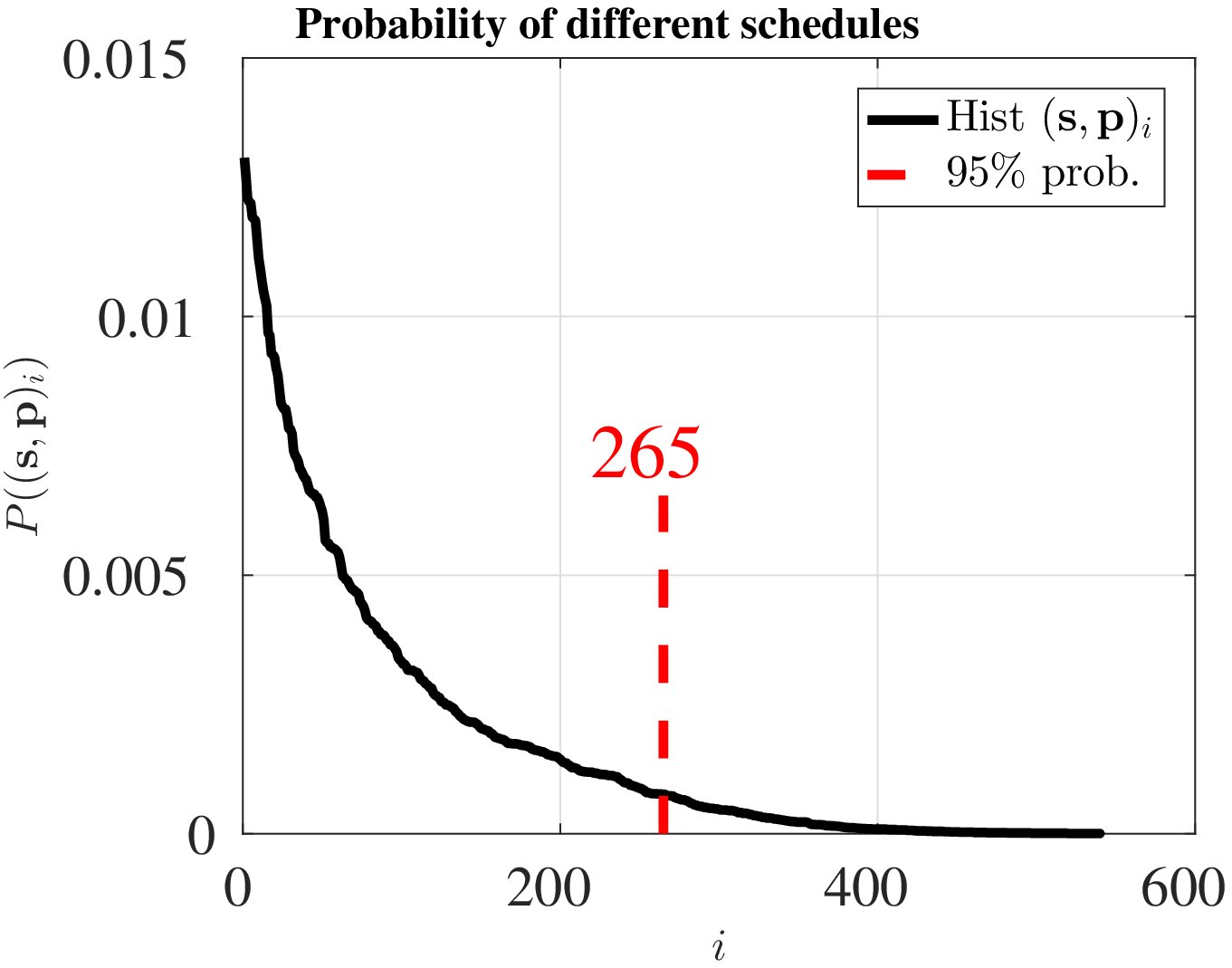}
  \label{fig:histRef}
  }
  \subfigure[SFWBF]{
    \includegraphics[width=0.3\columnwidth]{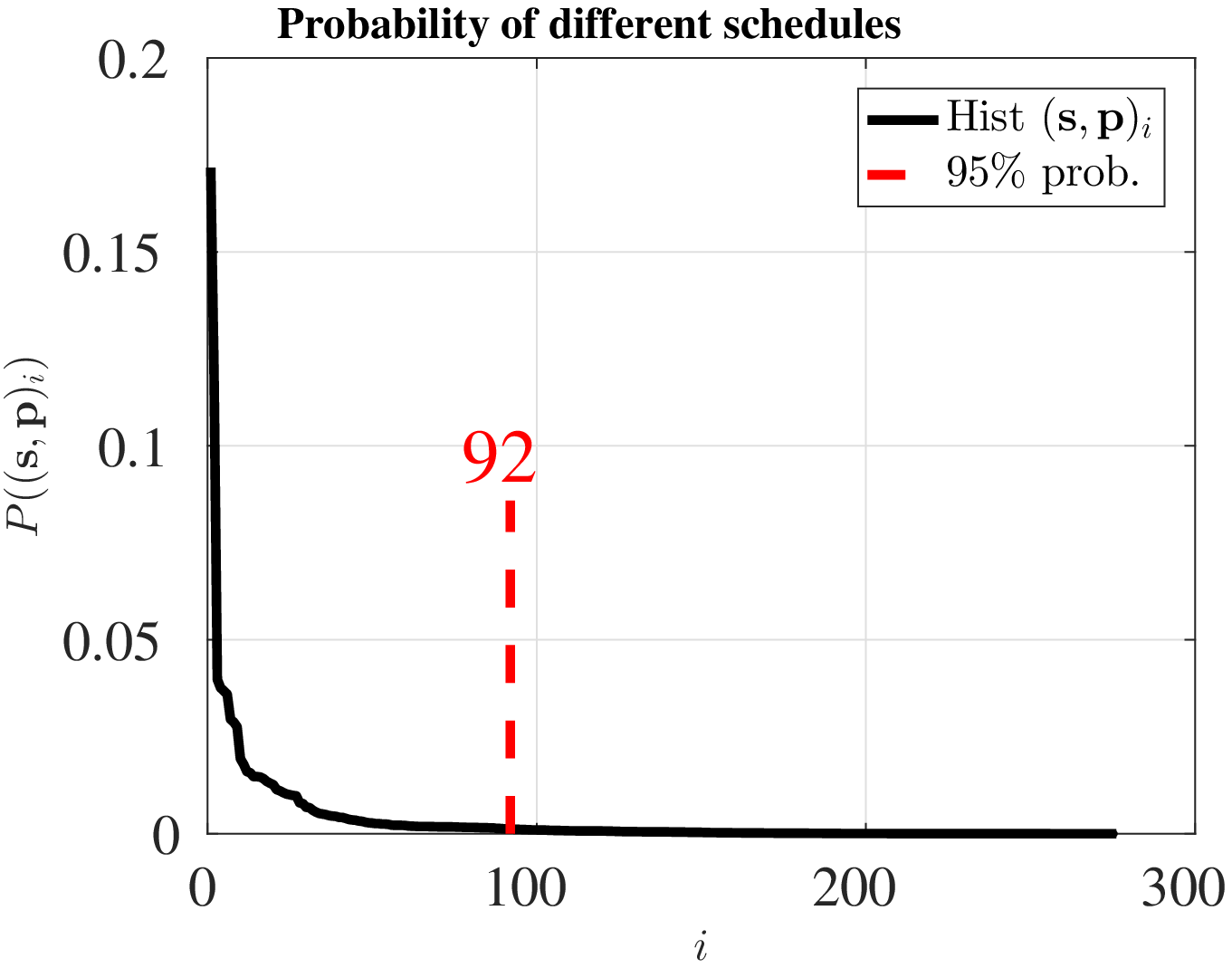}
  \label{fig:histSFWBF}
  }
  \subfigure[MFWLINSP]{
    \includegraphics[width=0.3\columnwidth]{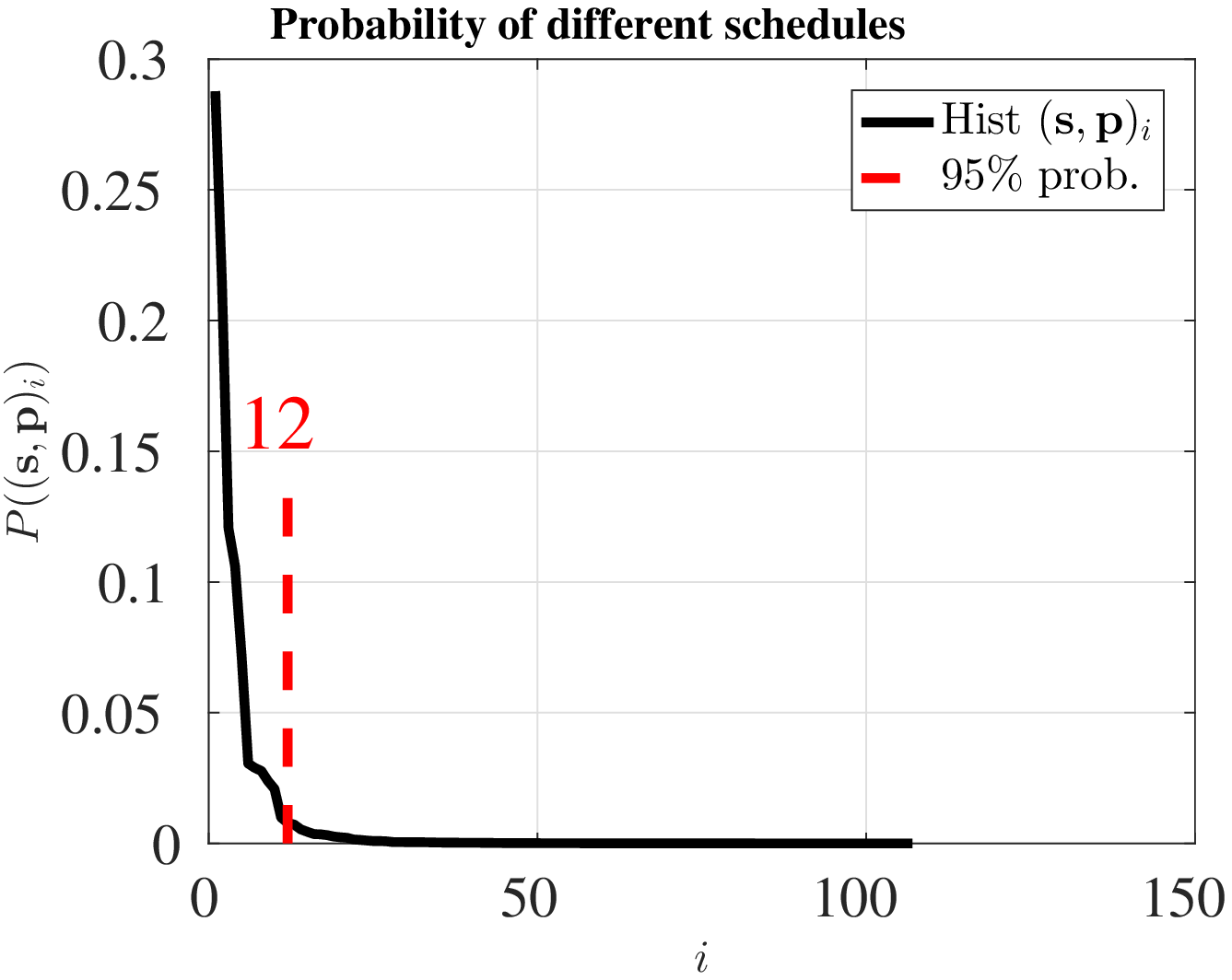}
  \label{fig:histMFWLINSP}
  }
  \subfigure[MFWMPSP]{
    \includegraphics[width=0.3\columnwidth]{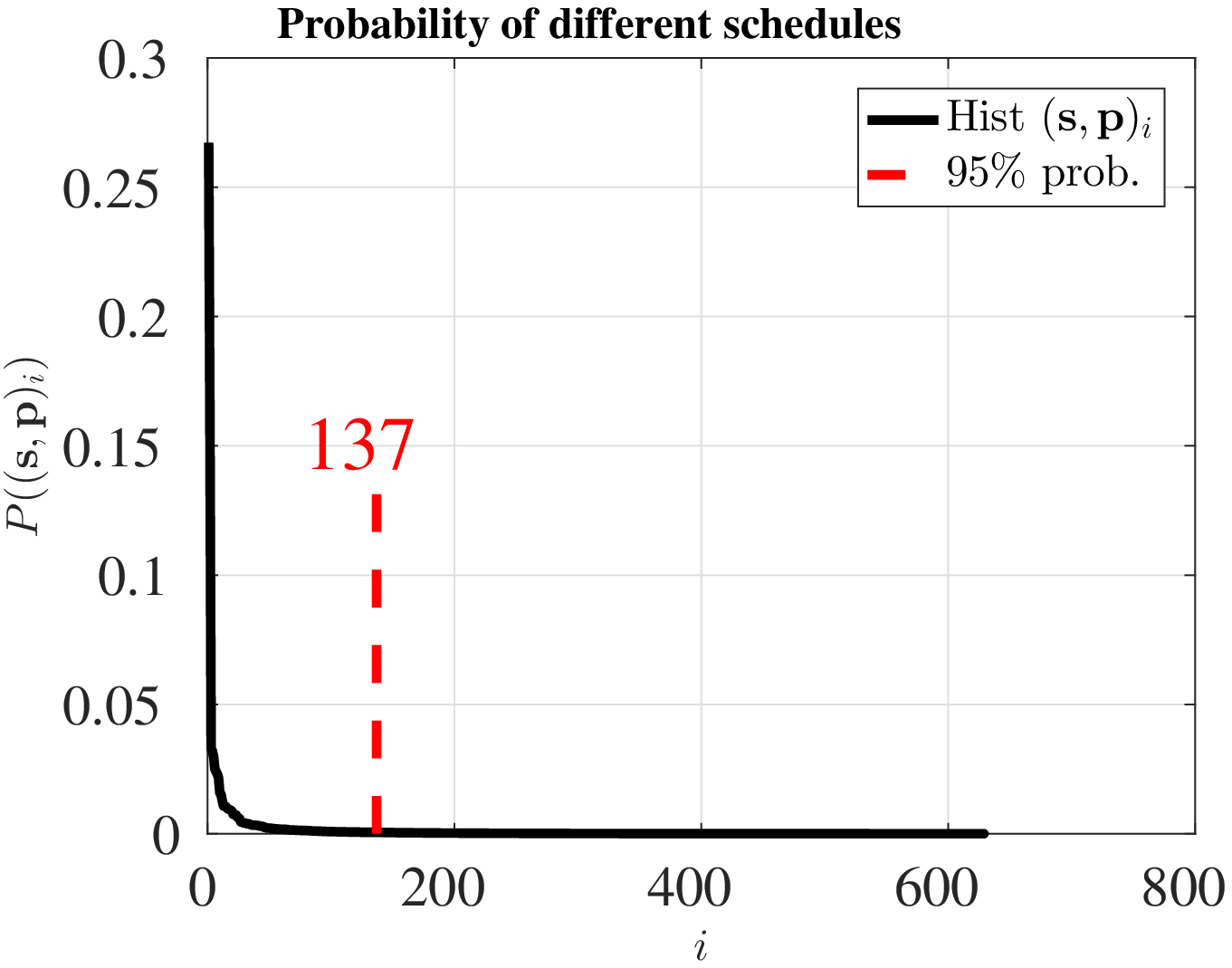}
  \label{fig:histMFWMPSP}
  }
  \subfigure[MFWMP]{
    \includegraphics[width=0.3\columnwidth]{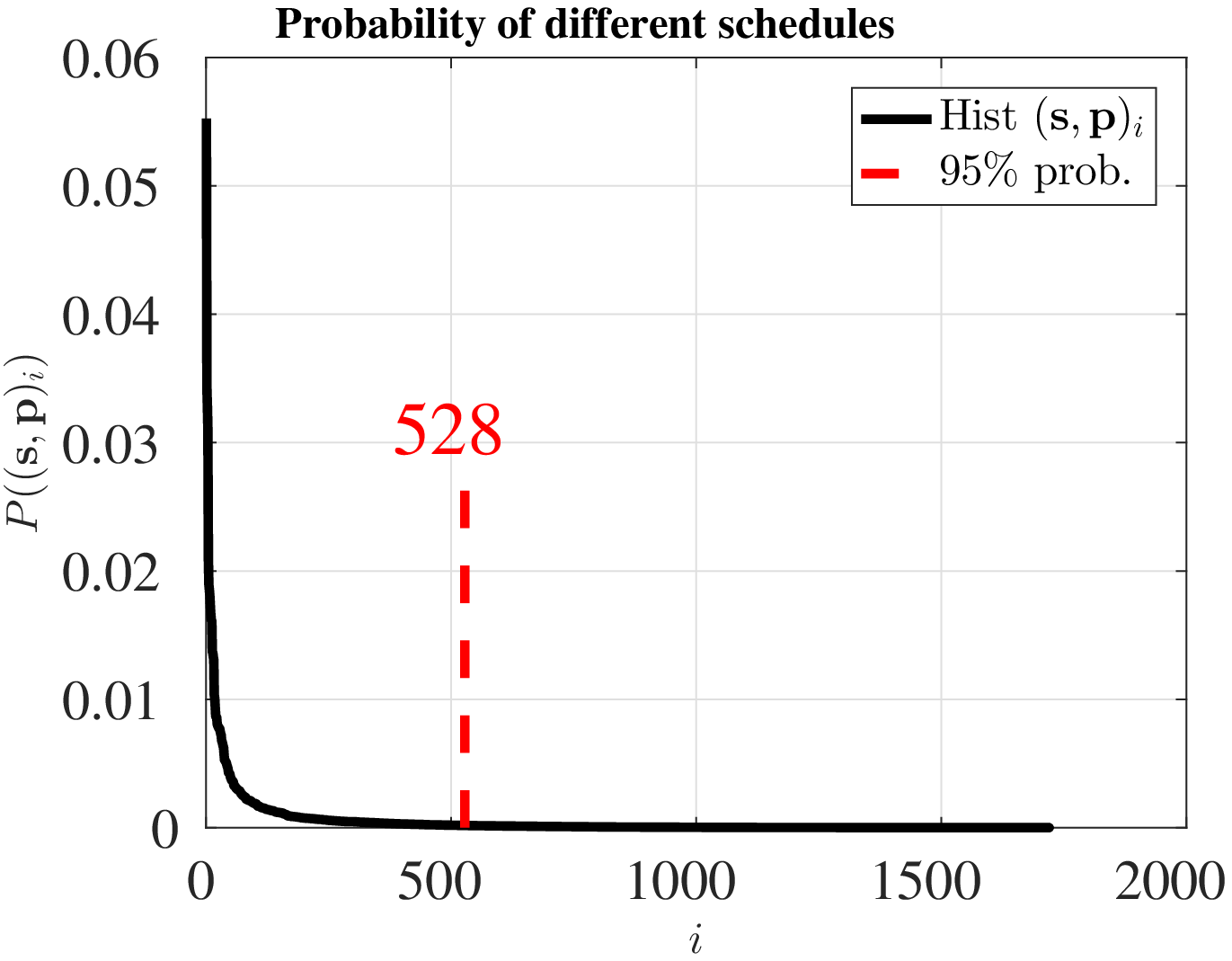}
  \label{fig:histMFWMP}
  }
  \subfigure[MFWPAC]{
    \includegraphics[width=0.3\columnwidth]{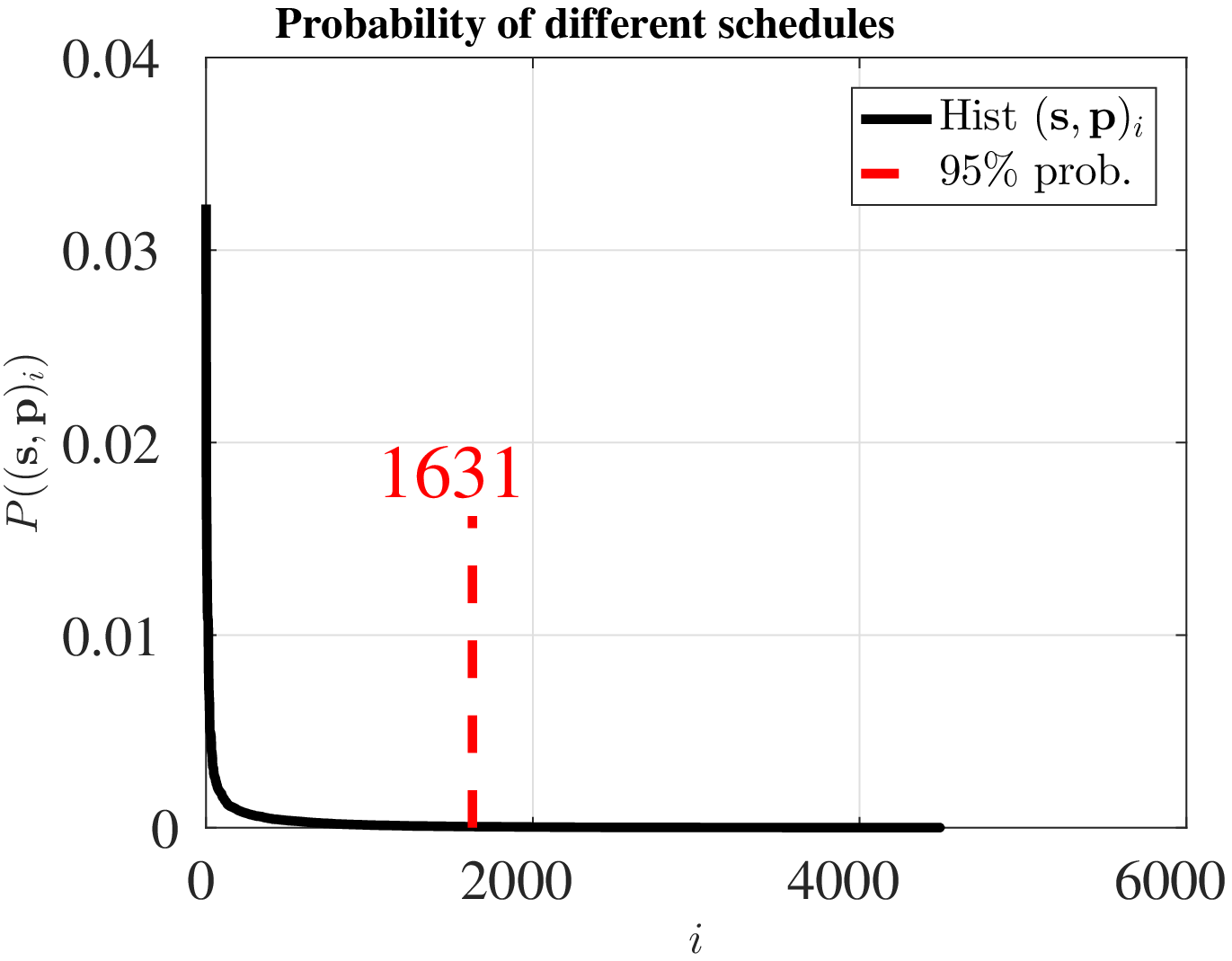}
  \label{fig:histMFWPAC}
  }
  \caption{Histogram of link schedules with different multiple-flow algorithms.}
  \label{fig:histograms}
\end{figure*}

\subsection{Proportional Fairness}

All algorithms employ the utility-maximization congestion control technique with a utility function $\frac{1}{2}\log(r)$\iftoggle{TRport}{{\iftoggle{TRportMARK}{\color{red}}{}. This function gives diminishing returns to an increase in rate; that is, the function values more when $2$ users receive $1$ Gbps each than when $1$ user receives $2$ Gbps, and so on. Among all diminishing-returns functions, the logarithm is of interest because it achieves the so-called proportional fairness, consisting in giving more resources to the users with a better channel and fewer resources to the users with a worse channel, but}}{,} maintaining a proportionality and guaranteeing the service of all users.

In Fig. \ref{fig:propfair} we show the throughput separated by user for three algorithms. We observe that the distribution of traffic is fairly similar in all algorithms, and proportional fairness is maintained
\iftoggle{TRport}{{\iftoggle{TRportMARK}{\color{red}}{}
  at least in a qualitative sense, although the numbers vary slightly.
}}

\begin{figure}
 \centering
 \includegraphics[width=0.5\columnwidth]{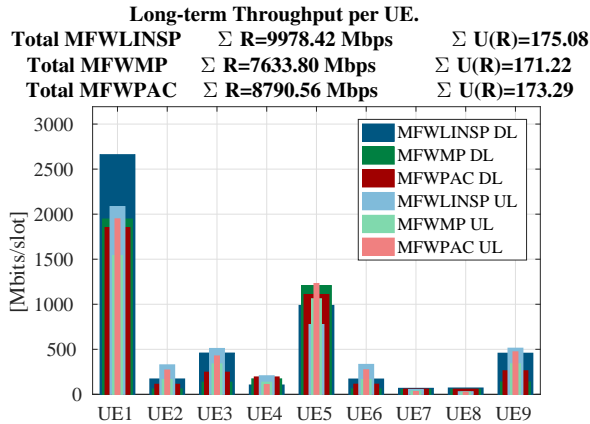}
 \caption{Long-term average throughput of each UE application over 200000 frames for three algorithms.}
 \label{fig:propfair}
\end{figure}

\iftoggle{TRport}{{\iftoggle{TRportMARK}{\color{red}}{}
  \subsection{Practical application}

  It must be noted that all algorithms in Fig. \ref{fig:propfair} are ``practical'' to a certain degree, in the sense that the scheduling policies could potentially be implemented in a real network, although due to the fact that NUM is a long-term result without short term guarantees, the usefulness of such implementation is limited to networks with no mobility or delay requirements.

  In the figure, the MFWLINSP algorithm achieves higher network utility and sum-rate than the PaC and MP implementation variants. This gain, however, would come at the expense of a centralized network controller that can access all network state information; and also of renouncing to the potential gains of waterfilling-based optimal power allocation. The next algorithm in terms of utility and rate is MFWPaC, which would allow a distributed implementation with optimal water-filling power allocation. However, the random algorithm would heavily penalize delay due to its highly unpredictable short-term behavior, as pointed out in Fig. \ref{fig:histograms}. Finally, MFWMP offers the least utility and rate, but it is a deterministic distributed algorithm that does not sacrifice power allocation opportunities.

  Unfortunately, the observations above apply only to one particular network realization and are not really universal results. In the next section, we study the average NUM over a large number of random network layout realizations, or ``drops.''
}}

\subsection{Consistency of the Algorithms}
\iftoggle{TRport}{{\iftoggle{TRportMARK}{\color{red}}{}
  We generalize the observations above, specific to one network topology, by repeating the study over 50 randomly generated node locations (``drops'') and averaging the results.
}}

Some of the algorithms select the optimal schedule under a set of constraints for the network (MFWLINOP, MFWLINSP, SFWBF, MWM). Comparing these algorithms, an increase in device power or scheduling flexibility is guaranteed to produce an improvement in network utility and throughput. We say that these algorithms are \textit{consistent}, because the relation between them is the same across all realizations of the random network.

In Fig. \ref{fig:consistent} we illustrate the achieved long-term average sum rate of each algorithm. As can be seen, MFWLINOP always outperforms SFWBF, which is always better than MWM. 

\begin{figure}
\centering
    \includegraphics[width=0.45\columnwidth]{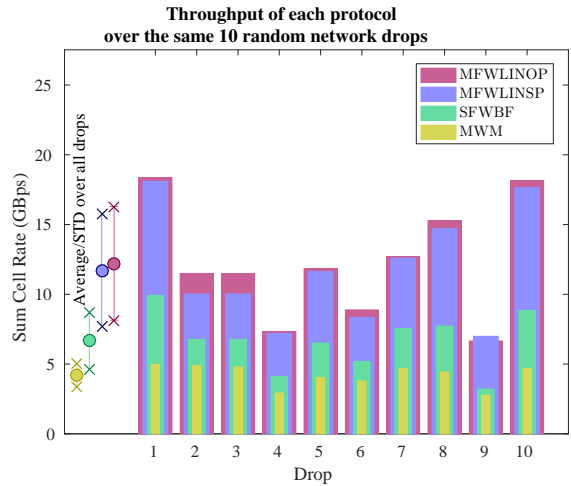}
  \caption{Sum rate of the more consistent algorithms.}
  \label{fig:consistent}
\end{figure}

We say that the 
\iftoggle{TRport}{{\iftoggle{TRportMARK}{\color{red}}{}
  algorithms MFWMPOP, MFWMP, MFWMPSP, MFWPAC and SFWMP
}}{other algorithms}
are \textit{inconsistent}\iftoggle{TRport}{{\iftoggle{TRportMARK}{\color{red}}{}
  \ due to the fact that the comparison between them is not guaranteed to manifest always in the same way. For instance, MFWMP has a better spatial multiplexing than SFWMP, but both algorithms are suboptimal MP schemes that can sometimes perform poorly in certain networks. It would be desirable that MFWMP always outperformed SFWMP due to its increased multiplexing, but in some certain topologies allowing SDM turns out to degrade the performance of MP schedulers, instead of improving it, and MFWMP is worse than SFWMP. In Fig. \ref{fig:inconsistent} we observe the long-term average sum rate of the inconsistent algorithms. Notice how, in drop 5, MFWMP and MFWPAC achieve more rate than MFWMPOP despite the latter having higher power in the nodes. Similarly, in drops 9 and 10, SFWMP achieves better rate than MFWMP despite the higher spatial multiplexing flexibility of the later.

  \begin{figure}
  \centering
    \includegraphics[width=0.45\columnwidth]{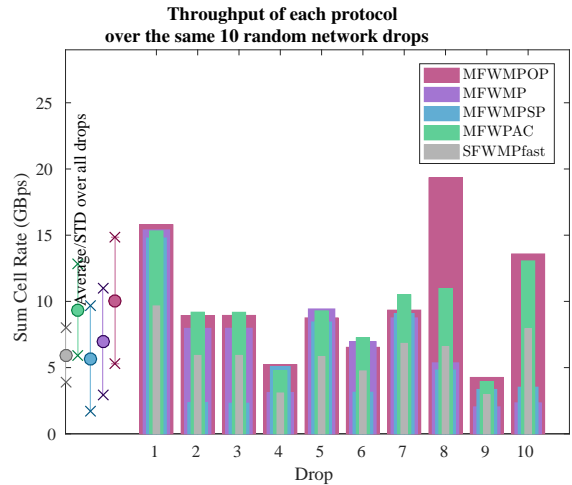}
    \caption{\textcolor{\iftoggle{TRportMARK}{red}{black}}{Sum rate of the more inconsistent algorithms.}}
    \label{fig:inconsistent}
  \end{figure}
  
}}{.}
All in all, even though the inconsistent algorithms can perform poorly in some unfortunate network topologies, they usually perform very well in many networks and should not be discarded that easily. In the next section we take 
a closer look at the performance of all algorithms averaged over many random networks, in order to understand the average and variance of the gains of each algorithm.

\subsection{Average Network Utility and Throughput Capacity}

We first performed a comparison among the MWM, SFWMP, MFWMPSP, MFWPAC, MFWLINOP and MFWLINSP algorithms. Figure \ref{fig:protoAverages} shows the maximum network utility and long-term sum throughput achieved by each algorithm in 50 random network realizations, and the corresponding average and standard deviation. We can observe two clear spatial-multiplexing improvements, from one-to-one transmission (MWM) to SDMA only (SFWMP) and from there to full spatial multiplexing (all algorithms with labeled prefix MFW-). This is the single most important conclusion of our paper: MU-MIMO techniques, if properly managed by scheduling, hold great potential for remarkable gains in throughput capacity and utility in a multi-hop mmWave networks. Additional research is needed to properly design practical scheduling and MAC protocols able to achieve these gains.

To study the performance of the SFWBF and MFWMP algorithms, whose computation complexity is very high, we select a random sample of 14 networks, and for each specific network instance we examine their performance and compare it with the SFWMP and MFLINOP benchmarks. Figure  \ref{fig:protoSamples} shows that, in each network instance, the performance of SFWMP is always close to that of SFWMP, the optimal SDMA-only scheduler. MFWMP beats SDMA-only protocols and performs closest to the full spatial multiplexing upper bound, MFWLINOP, in 8 of the 14 networks, whereas the MP heuristic shows a gain due to full spatial multiplexing (MFWMP beats SFWMP) in 11 of the 14 networks. These observations suggest that in a large number of network cases the exploitation of full spatial multiplexing can increase throughput in multi-hop networks even using heuristic MP distributed algorithms.

%
%

\iftoggle{TRport}{{\iftoggle{TRportMARK}{\color{red}}{}
  In addition, we have introduced different implementations of the MBP scheduler. The use of Split Power and distributed Message Passing produces an interesting average gain, but due to the algorithm inconsistency some networks are strongly benefited while others are penalized. The use of power allocation with random PaC produces an even greater gain, but this comes at the cost of less short-term guarantees in the network. Finally, MFWLINOP serves as an upper bound to the performance of any scheduling algorithm by the use of increased node transmit power, but violates the transmitter-power budget constraint.
}}

\begin{figure}
\centering
  \subfigure[Sum Long-term Throughput]{
  \iftoggle{TRport}{
    \includegraphics[width=0.45\columnwidth]{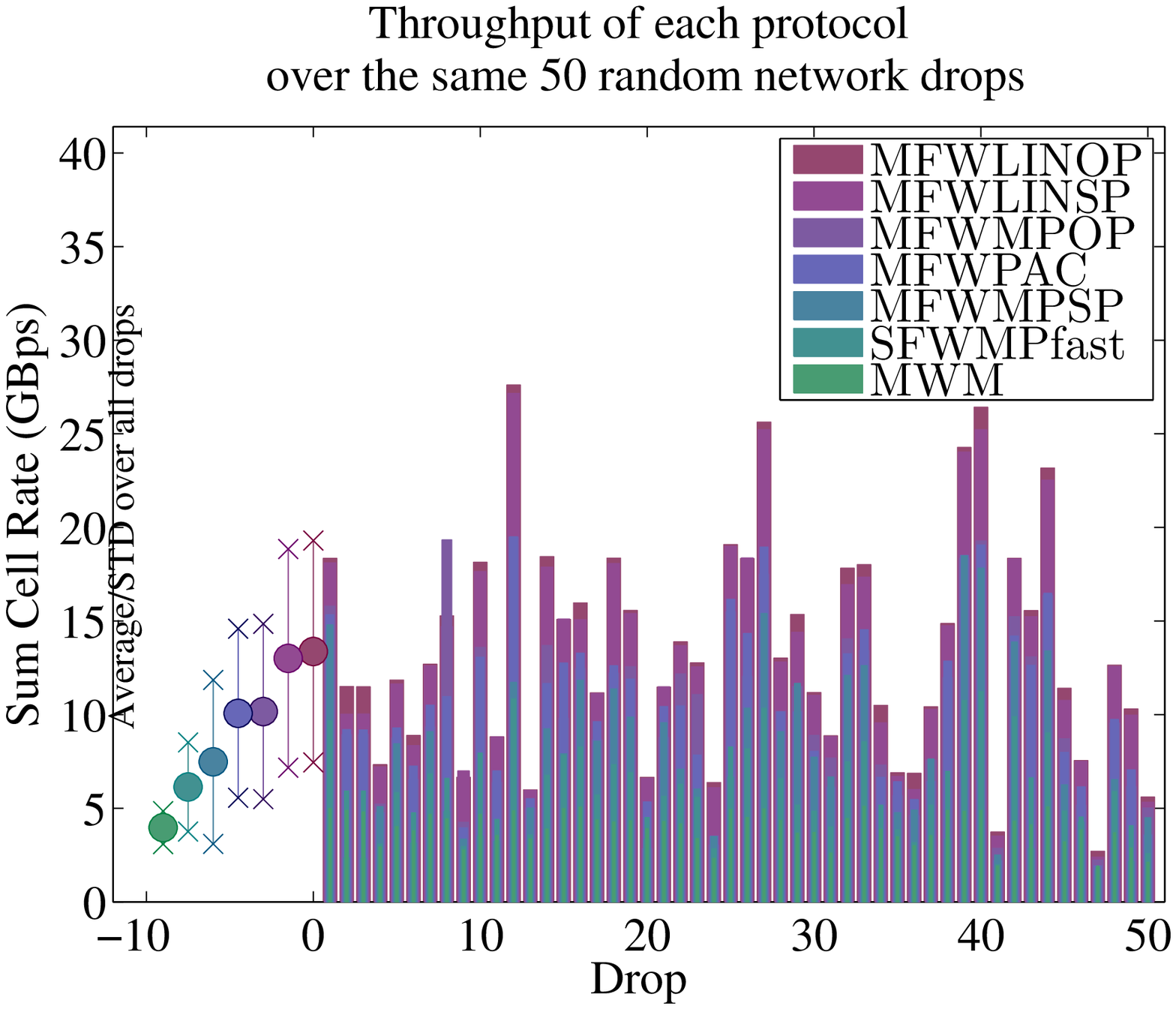}
    }{
    \includegraphics[width=0.45\columnwidth]{GMBP_sumRate_50_noreport}
    }
    \label{fig:sumRavg}
    }
  \subfigure[Sum. Network Utility]{
  \iftoggle{TRport}{
    \includegraphics[width=0.45\columnwidth]{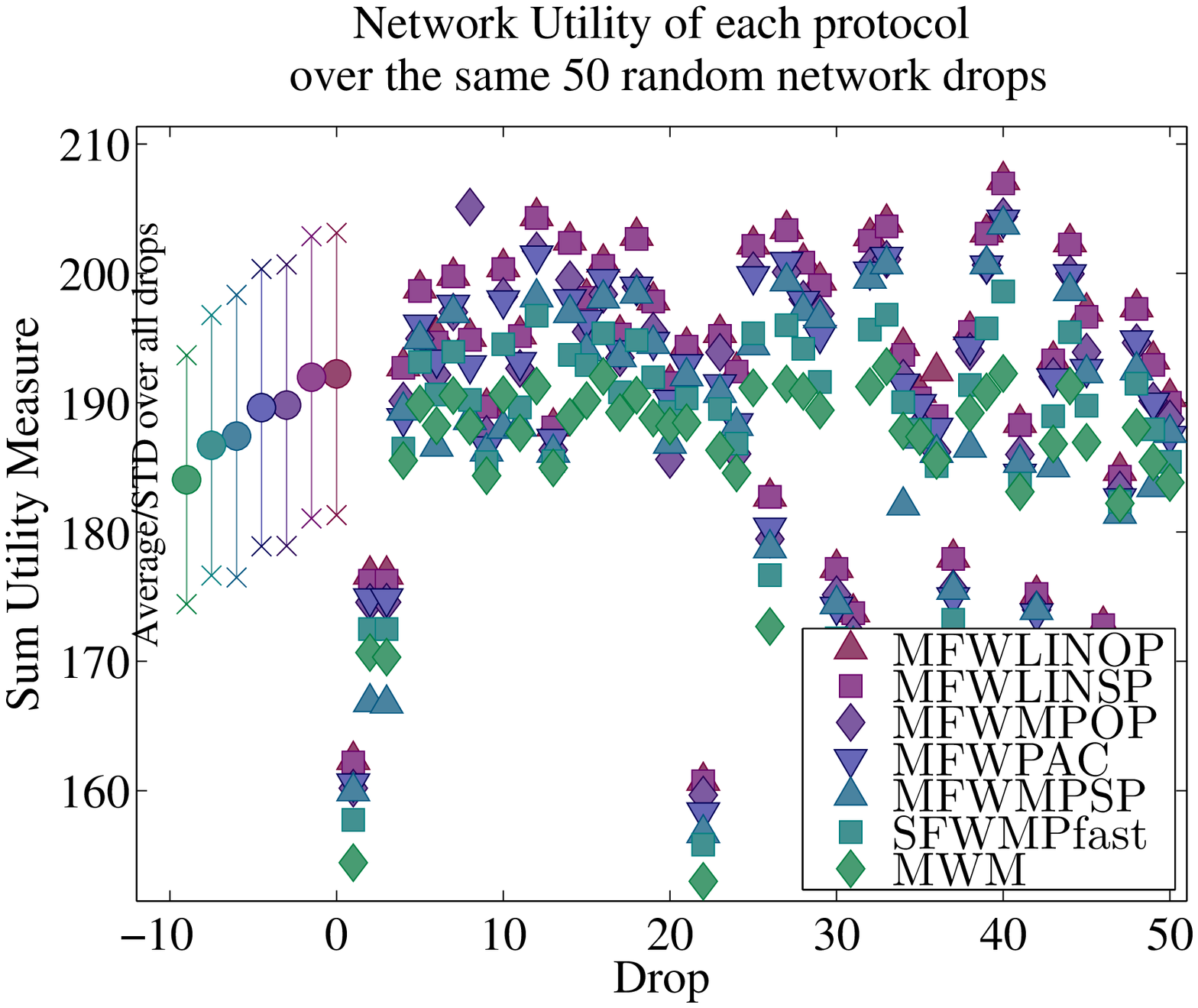}
    }{
    \includegraphics[width=0.45\columnwidth]{GMBP_utilityAll_50_noreport}
    }
    \label{fig:sumUavg}
    }
  \caption{Sum rate and utility of the algorithms over 50 drops.}
  \label{fig:protoAverages}
\end{figure}


  \begin{figure}
  \centering
    \subfigure[Sum Long-term Throughput]{
      \includegraphics[width=0.45\columnwidth]{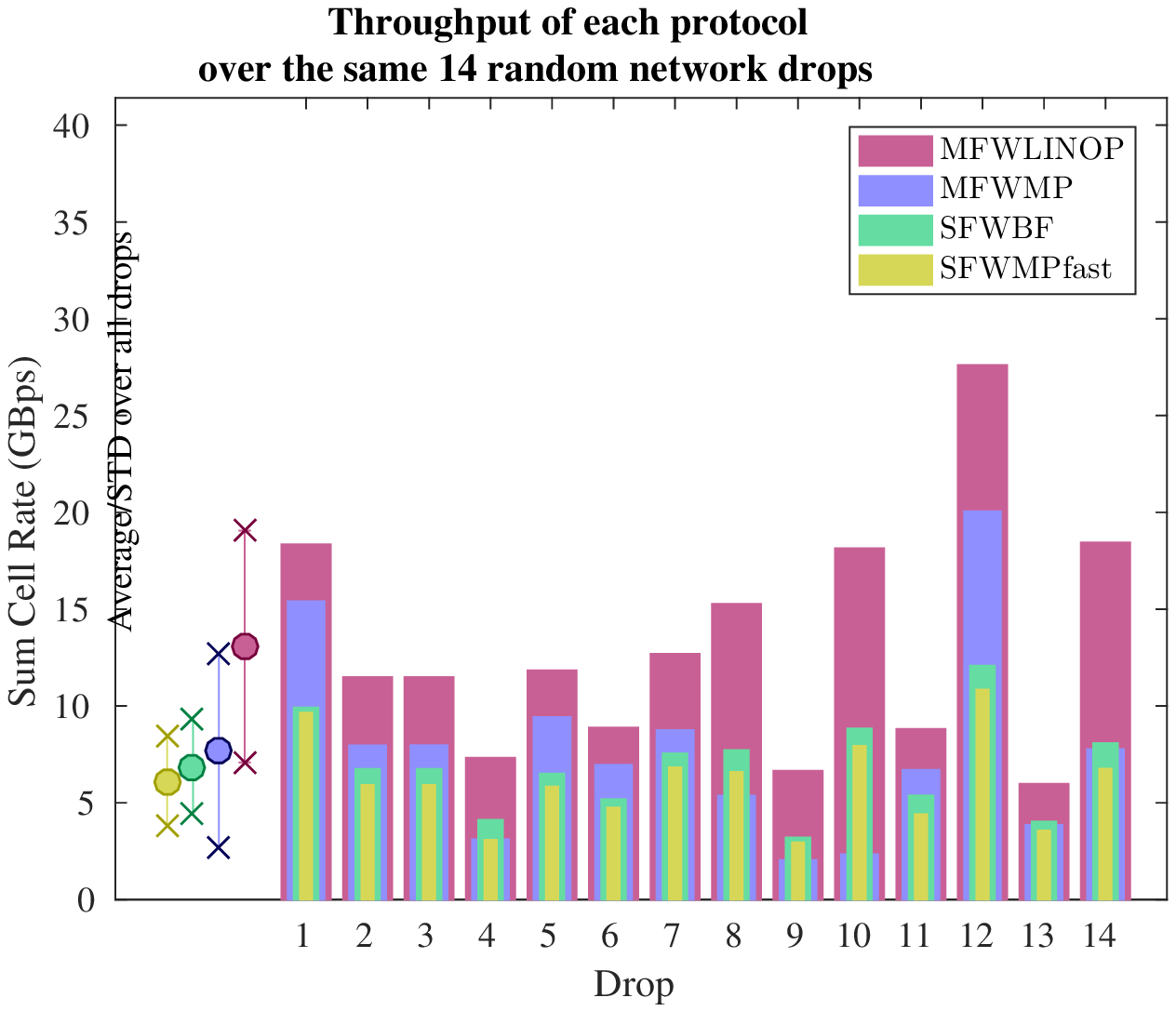}
      \label{fig:sumRsample}
      }
    \subfigure[Sum. Network Utility]{
      \includegraphics[width=0.45\columnwidth]{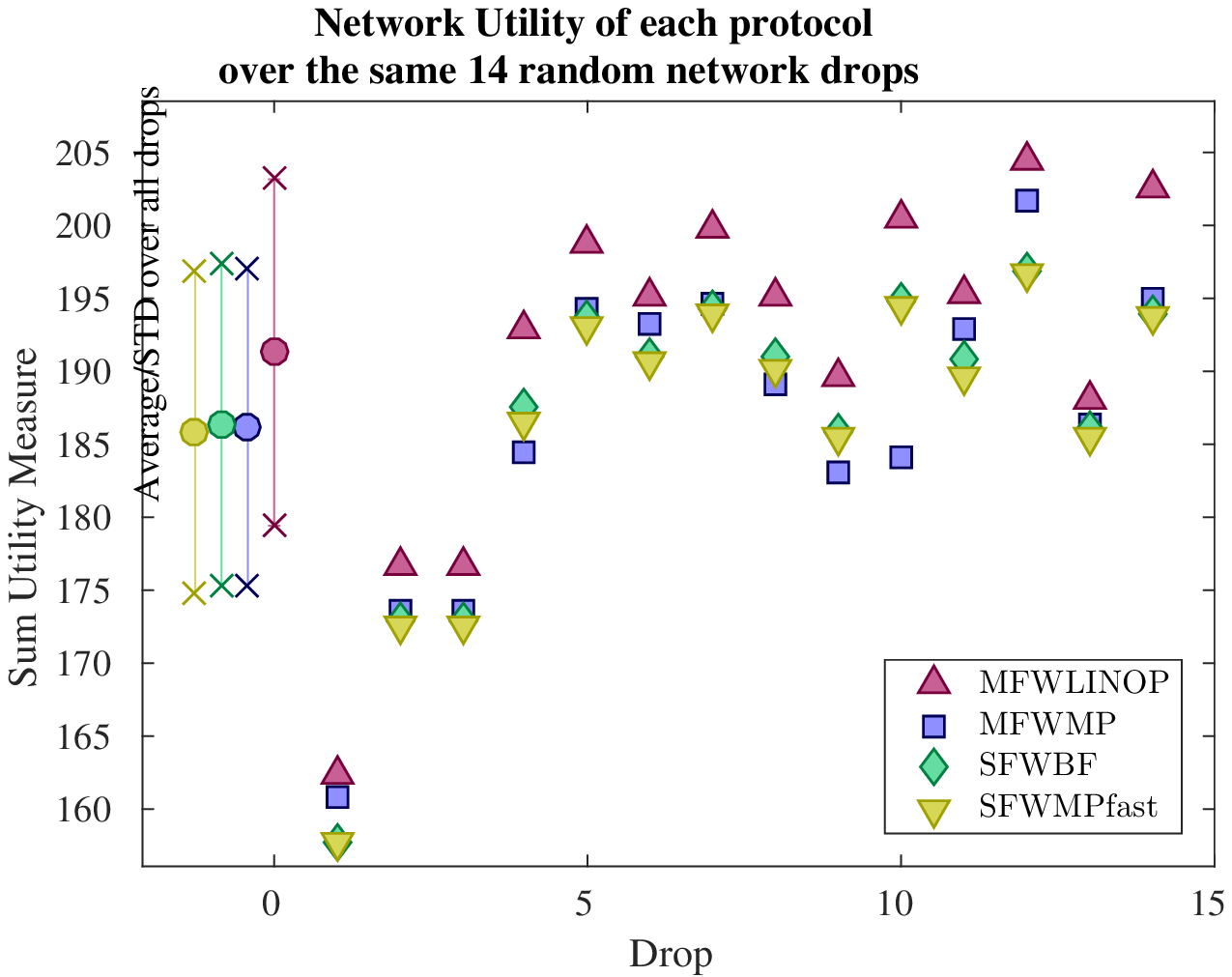}
      \label{fig:sumUsample}
      }
    \caption{Sum rate and utility of the algorithms over 14 drops.}
    \label{fig:protoSamples}
  \end{figure}

\subsection{Contribution Overview}

Among the algorithms in Table \ref{tab:algorithms}, there are five novel algorithms in this paper that have interesting applications

\begin{itemize} 
 \item MFWLINOP is a theoretical upper bound to NUM optimal performance that can be calculated but breaks the power constraints of the network.
 \item MFWLINSP is a practical NUM optimal algorithm with simplified fixed power allocation, though it requires centralized control. The results obtained with this algorithm are very close to those for MFWLINOP, which shows that MFWLINSP is near optimal and MFWLINOP is a tight upper bound.
 \item MFWPAC is a practical NUM optimal algorithm that can be distributed, although its randomness may affect delay.
 \item MFWMP is a practical suboptimal algorithm with optimal power allocation, that can be distributed and offers less randomness than PaC, but suffers catastrophic drops in performance in some networks where it works even worse than SFWMP, a protocol with a theoretically smaller capacity region.
 \item MFWMPSP is a practical suboptimal algorithm that exploits SDM with simplified fixed power allocation, can be distributed, and even though its performance drops in some networks, at least it is never below the results of SFWMP.
\end{itemize}

All algorithms outperform the MWM reference without MU-MIMO \cite{juanScheduling} and, with the exception of some realizations of MFWMP, the SFWMP reference with SDMA but no SDM \cite{gomezITAoptimal}.

\section{Conclusions and Future Work}
\label{sec:conclusion}

Future mmWave 5G networks will require a combined framework drawing from existing models in both single-hop cellular and multi-hop network paradigms. In the past, one driver of cellular rate increase has been the spatial multiplexing gain derived from MU-MIMO techniques. The physical layer assumptions contained in traditional multi-hop literature are usually very different from this. It is possible to adapt the classic multi-hop MBP techniques to this non-trivial PHY by separating the problems of role-assignment to the nodes, between the roles of transmitters and receivers, from the SDM and SDMA transmission design and power allocation, that is local to each specific transmitter thanks to the suppression of most interference due to the severe attenuation, blockage and mismatched array gain in mmWave. The scheduling problem is converted from a weighted matching on a graph to an optimal partition with the weight of each possible partition given by the solution to the power allocation. A formal proof of throughput-optimality and NUM can be obtained with minor adaptation of the classic proof for PaC.

In this paper we only consider interference-free links, where the spatial isolation of mmWave antenna arrays is assumed sufficient to make interfering signals buried in noise. We model multi-hop scheduling with MU-MIMO, enabling both SDMA at the receivers and SDM at the transmitters. This leaves a power allocation problem that can be solved independently at each transmitter using waterfilling or simplified to fixed power allocations if desired. We introduce multiple algorithms to obtain or approximate the optimal transmitter/receiver role assignment: random PaC, distributed heuristic MP, and centralized MILP (with fixed power only). We compare the results obtained with all algorithms. There are still many ways this model can be refined, but even using only heuristics we have observed significant throughput and network utility gains, which make a very strong case for research of multi-hop network scheduling incorporating MU-MIMO. Our future plans include the analysis of power allocation in the presence of interference, where the power allocation can no longer be considered separately at each transmitter, the study of interference control techniques in the scheduling model, and the adaptation of other optimization frameworks to address the MBP problem.

\iftoggle{TRport}{{\iftoggle{TRportMARK}{\color{red}}{}

\appendices
\section{Proof of NUM and Throughput Optimality}
\label{app:PAC}

Proofs for results very similar to Proposition \ref{pro:PAC} and \ref{pro:PACNUM} appear often in literature, considering different variants of the scheduling problem. Some examples include \cite{Tass,Eryilmaz2007,Kelly1997,Kelly1998,ModianoPower}. The main argument of the proof is always the same and traces back to the analysis of ergodic Markov chains by Tweedie \cite{Tweedie83Markov}. 

We combine ideas from two modifications of this proof with different network assumptions: in \cite{Eryilmaz2007,Eryilmaz2010Implementation} the throughput optimality is demonstrated for multi-hop network with an arbitrary interference model scheduling constraint, but only for constant-capacity links $c_{n,m}=1$ (links transmit one packet per frame, if two links in a vicinity transmit at the same time, they ``collide'' and both packets are lost, and the sets of links subject to such conflict is a given parameter). On the other hand, in some other works the proof is given for random-capacity links, but only for single-hop networks and only with a specific type of 2-hop interference constraint (a particular choice of conflict as a given parameter in the collision model). We have that these collision models do not accurately model directive transmissions and spatial multiplexing in mmWave networks, so we have replaced the collision model with a transmit/receive role assignment $\s(t)$ and a power allocation scheme $\pp(t)$. Moreover, we combine ideas from the multi-hop and variable-capacity proofs to create our version of the Proposition.

\subsection{General Overview}

The proof begins by considering the joint variable $\y(t)=(\q(t),\C(t))$ to represent a state of the network and scheduling system. This joint variable follows a Markov chain with some state space $\mathcal{Y}$. 

We consider a Lyapunov function of the state of the system defined as 
\begin{equation}
\mathcal{L}(\y(t))=\underset{\mathcal{L}_1(t)}{\underbrace{\sum_{n,f}|q_{n}^{f}(t)|^2}}+\underset{\mathcal{L}_2(t)}{\underbrace{\sum_{n,m,f}\left[([c_{n,m}^f]_{\textnormal{MBP}}-c_{n,m}^f)(q_{n}^{f}(t)-q_{m}^{f}(t))\right]^2}}
\end{equation}

The first term bounds the total queue length in the system
$$\mathcal{L}_1(t)=|\q(t)|^2=\q^T(t)\q(t)\geq \sum_{n,f} q_{n}^{(f)}(t)$$

The second term bounds the effect of the difference in effect between the random PaC scheduler and the ideal solution to \eqref{eq:MBP}. For convenience, we define a notation to measure the different weighted back pressure in the random scheduling and the optimal MBP scheduling. Let us denote the weight of a selected schedule by $w(t)=\q^T(t)(\C(t)+\C^T(t))\one$ and respectively the MBP by $w_{\textnormal{MBP}}(t)=\q^T(\C_{\textnormal{MBP}}(t)-\C_{\textnormal{MBP}}^T(t))\one$. The difference between the two being $\Delta w(t)=w_{\textnormal{MBP}}(t)-w(t)$ so the second term of the Lyapunov function is
$$ \mathcal{L}_2(t)=(\Delta w(t))^2=\left[\q^T(t)\left[\C_{\textnormal{MBP}}(t)-\C_{\textnormal{MBP}}^T(t)-\C(t)+\C^T(t)\right]\one\right]^2$$
which will be useful to characterize a part of the Lyapunov function and becomes zero when MBP is used.

The goal is to show that under the condition that there is a finite highest link capacity, $C_{\max}\triangleq \max c_{n,m}(t)<\infty$, it can be proved that the average Lyapunov drift 
$$\mathcal{D}(\y(t))=\Ex{t}{\mathcal{L}(\y(t+1))-\mathcal{L}(\y(t))|\y(t)}$$
is always negative if $\mathcal{L}(\y(t))$ is greater than some large scalar $B$. 

When this claim is satisfied, the Foster-Lyapunov criterion \cite{Tweedie83Markov} establishes that the Markov process is possitive recurrent with a steady state distribution contained in $\mathcal{S}=\{\y(t): \mathcal{L}(\y(t))<B\}\subset\mathcal{Y}$. Due to the fact that we define a Lyapunov function that upper bounds the aggregate queue length in the system, the states $\mathcal{S}$ correspond to a stable network by definition. In addition, due to the fact that the Lyapunov function grows with the difference between MBP and PAC, the steady states contained $\mathcal{S}$ are also associated with a small difference in long-term throughput between the two schedulers.

The interpretation of this result is that the system with sufficiently long queues has a stochastic tendency to drift to lower queue length states, and simultaneously the states that select a schedule that is very far off of the optimal have a stochastic tendency to select better schedules next; both behaviors occurring at the same time stabilize the network and guarantee NUM.

\subsection{Analysis of $\mathcal{L}_1(t)$}

We compute the 1-step Lyapunov drift of the function $\mathcal{L}_1(t)$ as
$$\mathcal{D}_1(t)=\Ex{}{\mathcal{L}(\y(t+1))-\mathcal{L}(\y(t))}$$

Hereafter, we introduce the queue-update function \eqref{eq:qupdate} and we drop the time index $(t)$ . This leaves
$$\mathcal{D}_1(t)=\Ex{}{|\q+(\C^T-\C)\one_{NF,1}+\ab|^2-|\q|^2}$$

Expanding the square sum we get
\begin{equation}
\begin{split}
   \mathcal{D}_1(t)&=\underset{ 0}{\underbrace{\Ex{}{|\q|^2-|\q|^2}}}+\underset{\leq N^2C_{\max}+N\Omega_{\max}C_{\max}}{\underbrace{\Ex{}{|(\C^T-\C)\one_{NF,1}+\ab|^2}}}\\
   &+2\Ex{}{\q^T\left[(\C^T-\C)\one_{NF,1}+\ab\right]}\\
  \end{split}
\end{equation}

We add and subtract the term $\Ex{}{2V\one_{1,NF}\mathcal{U}(\ab)}$ (this can be omitted for throughput optimality with a fixed rate $\lambda\in\Lambda$), wrap the first terms in a constant $C_1$ and reorganize the last term
\begin{equation}
\begin{split}
   \mathcal{D}_1(t)&=C_1+\Ex{}{2V\one_{1,NF}\mathcal{U}(\ab)}\\
   &-\Ex{}{\Ex{}{2V\one_{1,NF}\mathcal{U}(\ab)}-2\q^T\ab}\\
   &+2\Ex{}{\q^T(\C^T-\C)\one_{NF,1}}\\
  \end{split}
\end{equation}

We next add and subtract $2\Ex{}{\q^T(\C_{\textnormal{MBP}}^T-\C_{\textnormal{MBP}})\one_{NF,1}}$, expressing the Lyapunov drift of the system as the drift of an optimal system plus a gap that is equivalent to $\Delta w(t)$
\begin{equation}
\begin{split}
   \mathcal{D}_1(t)&=C_1+\Ex{}{2V\one_{1,NF}\mathcal{U}(\ab)}\\
   &-\Ex{}{\Ex{}{2V\one_{1,NF}\mathcal{U}(\ab)}-2\q^T\ab}\\
   &+2\Ex{}{\q^T(\C_{\textnormal{MBP}}^T-\C_{\textnormal{MBP}})\one_{NF,1}}\\
   &+2\underset{\Ex{}{\Delta w(t)}=\Ex{}{\sqrt{\mathcal{L}_2(t)}}}{\underbrace{\Ex{}{\q^T(\C_{\textnormal{MBP}}-\C_{\textnormal{MBP}}^T)\one_{NF,1}-\q^T(\C-\C^T)\one_{NF,1}}}}
  \end{split}
\end{equation}

We introduce $\x^{V}$, the solution to the approximate problem \eqref{eq:EPSproblem}, which by definition maximizes the absolute value of the third term. The drift is averaged conditioned on $\y(t)$; on a given state $\y(t)$ the value of $\mathcal{L}_1(t)$ is deterministic, so we clear the $\Ex{}{}$ from that term, obtaining
\begin{equation}
\begin{split}
   \mathcal{D}_1(t)&\leq C_1+\Ex{}{2V\one_{1,NF}\mathcal{U}(\ab)}\\
   &-\Ex{}{\Ex{}{2V\one_{1,NF}\mathcal{U}(\x^{V})}-2\q^T\x^{V}}\\
   &+2\Ex{}{\q^T(\C_{\textnormal{MBP}}^T-\C_{\textnormal{MBP}})\one_{NF,1}}\\
   &+2\sqrt{\mathcal{L}_2(t)}\\
  \end{split}
\end{equation}
we then arrange this as 
\begin{equation}
\begin{split}
   \mathcal{D}_1(t)&\leq C_1+\Ex{}{2V\one_{1,NF}\mathcal{U}(\ab)-2V\one_{1,NF}\mathcal{U}(\x^{V})}\\
   &+2\Ex{}{\q^T(\C_{\textnormal{MBP}}-\C_{\textnormal{MBP}}^T)\one_{NF,1}-\q^T\x^{V}}\\
   &+2\sqrt{\mathcal{L}_2(t)}\\
  \end{split}
\end{equation}

From the assumption that $\x^{V}\in\Lambda$ we get that there exists some convex linear combination of feasible capacities $\C^V\in \text{Co}(\mathcal{C})$ such that the net traffic in the source is $((\C^V)^T-\C^V)\one_{NF,1}-\epsilon \one_{NF,1}=\x^{V}$ for some small $\epsilon$. By contradiction, if no such linear combination existed, then there would be at least one source queue that grows to infinity and $\x^{V}\notin\Lambda$ and the network would never be stable.

Finally, we note that $\min_{\pp} \q^T(\C^T-\C)\one_{NF,1}$ is equivalent to $ =\max_{\pp} \q^T(\C-\C^T)\one_{NF,1}$ and both are reorderings of the MBP in \eqref{eq:MBP}, therefore $\q^T(\C_\text{MBP}^T-\C_\text{MBP})\one_{NF,1}<\q^T((\C^V)^T-\C^V)\one_{NF,1}$.

\begin{equation}
\begin{split}
\label{eq:finalD1}
   \mathcal{D}_1(t)&\leq C_1+\Ex{}{2V\one_{1,NF}\mathcal{U}(\ab)-2V\one_{1,NF}\mathcal{U}(\x^{V})}\\
   &-\Ex{}{\epsilon \q^T\one_{NF,1}}+2\sqrt{\mathcal{L}_2(t)}\\
  \end{split}
\end{equation}

This captures the essential simplifications on specific terms in Propositions \ref{pro:MBP} and \ref{pro:NUMCC}. 
\begin{itemize}
 \item The second term is related to NUM and can be omitted if traffic is static and assumed to be in the throughput capacity region, $\lambdab\in \Lambda$, to prove throughput-optimality of an inelastic traffic systems. 
 \item The third term is related to the use of random or sub-optimal approximations of \eqref{eq:MBP}. If the MBP scheduler is directly employed, the term $\mathcal{L}_2(t)$ is zero and the proof is simplified. 
\end{itemize}

The second term can in general be upper bounded by $C_2=2VN\mathcal{U}(C_{\max}\Omega_{\max})$, which means that the Lyapunov drift can be written as a positive constant minus the sum queue length and scheduling weight differential
\begin{equation}
 \label{eq:D1compact}
 \mathcal{D}_1(t)\leq C_1+C_2-\epsilon |\q|_1+2\sqrt{\mathcal{L}_2(t)}
\end{equation}

The proof would be complete with this if this was the case with MBP scheduling, where the Foster-Lyapunov criterion would be met by defining a complementary to the set of steady states $\mathcal{S}^c$ where any sufficiently long queues $\q(t)$ satisfy $C_1+C_2-\epsilon |\q|_1<0$. In addition, since this stabilization converges to the approximate optimum $\x^{V}$, we have that the distance to the true optimum is $\one_{1,NF}\mathcal{U}(\x^*)-\one_{1,NF}\mathcal{U}(\x^{V})\leq \frac{C_1}{V}$. This suggests that $V\geq10C_{\max}^2$ gives close approximations of the optimal solution to the MBP subproblem \eqref{eq:MBP}.

Since we are interested in extending the proof to a random PaC scheduler, the following additional steps are necessary.

\subsection{Analysis of $\mathcal{L}_2(t)$+$\mathcal{L}_1(t)$}
We compute the Lyapunov drift of the complete function $\mathcal{L}(t)$ as
$$\mathcal{D}(t)=\Ex{}{\mathcal{L}(\y(t+1))-\mathcal{L}(\y(t))}=\underset{\mathcal{D}_1(t)}{\underbrace{\Ex{}{\mathcal{L}_1(\y(t+1))-\mathcal{L}_1(\y(t))}}}+\underset{\mathcal{D}_2(t)}{\underbrace{\Ex{}{\mathcal{L}_2(\y(t+1))-\mathcal{L}_2(\y(t))}}},$$
where $\mathcal{D}_1(t)$ follows the analysis steps above up to \eqref{eq:D1compact}, which shows the dependency of the queue drift on $\mathcal{L}_2(t)$ if the latter is not zero.

We focus therefore on the drift of $\mathcal{L}_2(t)$, which can be derived from squaring

\begin{equation}
  \begin{split}\Delta w(t+1)&=\Ex{}{\q^T(t+1)\left[\C_{\textnormal{MBP}}(t+1)-\C_{\textnormal{MBP}}^T(t+1)-\C(t+1)+\C^T(t+1)\right]\one}\\
			    &=(1-\delta)\Ex{}{\q^T(t+1)\left[\C_{\textnormal{MBP}}(t+1)-\C_{\textnormal{MBP}}^T(t+1)-\C(t+1)+\C^T(t+1)\right]\one|w(t+1)>0}        
    \end{split}
\end{equation}
where the conditioning on the second step comes from the fact that $w(t+1)$ is zero with probability $\delta$.

To upper bound the conditional mean we express $\q^T(t+1)=\q^T(t)+\Delta\q$ where $\Delta\q$ represents a vector of arrivals to the queues that is bounded. The details of $\Delta\q$ are irrelevant as long as it is bounded, but we could easily obtain it from \eqref{eq:qupdate}. This gives some interesting rearrangements

\begin{equation}
  \begin{split}
    \Delta w(t+1)=\;&\Ex{}{\Delta\q^T(t)\left[\C_{\textnormal{MBP}}(t+1)-\C_{\textnormal{MBP}}^T(t+1)\right]\one-\Delta\q^T(t)\left[\C(t+1)+\C^T(t+1)\right]\one}\\
    &+\Ex{}{\q^T(t)\left[\C_{\textnormal{MBP}}(t+1)-\C_{\textnormal{MBP}}^T(t+1)\right]\one-\q^T(t)\left[\C(t+1)-\C^T(t+1)\right]\one}\\
    \end{split}
\end{equation}

The first line of this expression can be bounded with constants due to the boundedness of the queue arrivals vector. The second line can be bounded by $\Delta w(t))=\sqrt{\mathcal{L}_2(t)}$ using two rules
\begin{itemize}
 \item $\q^T(t)\left[\C_{\textnormal{MBP}}(t+1)-\C_{\textnormal{MBP}}^T(t)\right]\one<\q^T(t)\left[\C_{\textnormal{MBP}}(t)-\C_{\textnormal{MBP}}^T(t)\right]\one$ due to the definition of $\C_{\textnormal{MBP}}(t)$
 \item $\q^T(t)\left[\C(t+1)-\C^T(t+1)\right]>\q^T(t)\left[\C(t)-\C^T(t)\right]$ due to the definition of the decision rule in PaC.
\end{itemize}

Using the above and the fact that there is a finite highest link rate in the network $C_{\max}\triangleq \max_{n,m} c_{n,m}(p_{n,m}=1)<\infty$, $\Delta \q$ is bounded and so
$$\Delta w(t+1)\leq (1-\delta)[C_{max}NF+\Delta w(t)]$$
With this we can upper bound the second term of the Lyapunov drift as
$$\mathcal{D}_2(t)\leq (C_{max}NF+\Delta w(t))^2-(\Delta w(t))^2 =C_3+C_4\sqrt{\mathcal{L}_2(t)}-(\delta-\delta^2)\mathcal{L}_2(t)$$

This means that we can write the combined drift of the system as 
$$\mathcal{D}(t)\leq  C_1+C_2+C_3-\epsilon |\q|_1+(2+C_4)\sqrt{\mathcal{L}_2(t)}-(\delta-\delta^2)\mathcal{L}_2(t)$$

Now define the set $\mathcal{S}$ where we want to evaluate the Foster-Lyapunov criteria as the set of states such that $\mathcal{L}(t)<B$ for some arbitrary large number $B$. We have that for all states in $\mathcal{S}^c$, either of two cases occur, depending on the comparison of the values of $\sqrt{\mathcal{L}_2(t)}$ and $\frac{(2+C_4)}{\delta-\delta^2}$ (note this is some large finite constant too)
\begin{itemize}
 \item If $\sqrt{\mathcal{L}_2(t)}>\frac{(2+C_4)}{\delta-\delta^2}$, all terms that grow with $\mathcal{L}_1(t)$ or $\mathcal{L}_2(t)$ in  $\mathcal{D}(t)$ are negative and the drift is negative for a sufficiently high $\mathcal{L}(t)>B$.
 \item On the other hand, if $\sqrt{\mathcal{L}_2(t)}<\frac{(2+C_4)}{\delta-\delta^2}$ but $\mathcal{L}(t)>B$, we can write $|\q|_1=\sqrt{\mathcal{L}_1(t)}\geq \sqrt{\mathcal{L}_1(t)}\geq \sqrt{B-\mathcal{L}_2(t)}$.  We divide the term $|\q|_1$ in two halves and rewrite the drift as
$$\mathcal{D}(t)\leq  C_1+C_2+C_3-\epsilon \frac{1}{2}|\q|_1-\epsilon\underset{\frac{1}{2}|\q|_1}{\underbrace{\frac{1}{2}\sqrt{B-\mathcal{L}_2(t)}}}+2\sqrt{\mathcal{L}_2(t)}-(\delta-\delta^2)\mathcal{L}_2(t)$$
Finally, we note that for any $0<\delta<1$ and $\sqrt{\mathcal{L}_2(t)}<\frac{(2+C_4)}{\delta-\delta^2}$ if we use sufficiently high $B$ we can guarantee that the term $-\epsilon\frac{1}{2}\sqrt{B-\mathcal{L}_2(t)}+2\sqrt{\mathcal{L}_2(t)}$ is negative.
\end{itemize}

Thus, the drift $\mathcal{D}(t)$ is negative for all states in $\mathcal{S}^c$ in both cases, the Foster-Lyapunov criterion is met, and the system is positive recurrent with a steady state distribution in $\mathcal{S}$ QED.

}}


\begin{thebibliography}{10}
\providecommand{\url}[1]{#1}
\csname url@samestyle\endcsname
\providecommand{\newblock}{\relax}
\providecommand{\bibinfo}[2]{#2}
\providecommand{\BIBentrySTDinterwordspacing}{\spaceskip=0pt\relax}
\providecommand{\BIBentryALTinterwordstretchfactor}{4}
\providecommand{\BIBentryALTinterwordspacing}{\spaceskip=\fontdimen2\font plus
\BIBentryALTinterwordstretchfactor\fontdimen3\font minus
  \fontdimen4\font\relax}
\providecommand{\BIBforeignlanguage}[2]{{%
\expandafter\ifx\csname l@#1\endcsname\relax
\typeout{** WARNING: IEEEtran.bst: No hyphenation pattern has been}%
\typeout{** loaded for the language `#1'. Using the pattern for}%
\typeout{** the default language instead.}%
\else
\language=\csname l@#1\endcsname
\fi
#2}}
\providecommand{\BIBdecl}{\relax}
\BIBdecl

\bibitem{gomezITAoptimal}
F.~G{\'{o}}mez-Cuba and M.~Zorzi, ``{Optimal link scheduling in millimeter wave
  multi-hop networks with space division multiple access},'' in \emph{IEEE
  Information Theory and Applications Workshop (ITA)}, 2016.

\bibitem{rappaport2013}
Y.~Azar, G.~Wong, K.~Wang, R.~Mayzus, J.~Schulz, H.~Zhao, F.~Gutierrez,
  D.~Hwang, and T.~Rappaport, ``{28 GHz Propagation measurements for outdoor
  cellular communications using steerable beam antennas in New York City},'' in
  \emph{IEEE International Conference on Communications (ICC)}, 2013.

\bibitem{RanRapEr:14}
S.~Rangan, T.~T.~S. Rappaport, and E.~Erkip, ``{Millimeter-wave cellular
  wireless networks: Potentials and challenges},'' \emph{Proceedings of the
  IEEE}, vol. 102, no.~3, pp. 366--385, 2014.

\bibitem{hoymann2012relaying}
C.~Hoymann, W.~Chen, J.~Montojo, A.~Golitschek, C.~Koutsimanis, and X.~Shen,
  ``{Relaying operation in 3GPP LTE: challenges and solutions},'' \emph{IEEE
  Communications Magazine}, vol.~50, no.~2, pp. 156--162, 2012.

\bibitem{parkvall2011evolution}
S.~Parkvall, A.~Furuskar, and E.~Dahlman, ``{Evolution of LTE toward
  IMT-advanced},'' \emph{IEEE Communications Magazine}, vol.~49, no.~2, pp.
  84--91, 2011.

\bibitem{fgomez2014improvedrelaying}
F.~G{\'{o}}mez-Cuba and F.~J. Gonz{\'{a}}lez-Casta{\~{n}}o, ``{Improving
  third-party relaying for LTE-A: A realistic simulation approach},'' in
  \emph{IEEE International Conference on Communications (ICC)}, 2014.

\bibitem{YoB:12}
B.~Yu, S.~Mukherjee, H.~Ishii, and L.~Yang, ``{Dynamic TDD support in the LTE-B
  enhanced local area architecture},'' in \emph{GC'12 Workshop: The 4th IEEE
  International Workshop on Heterogeneous and Small Cell Networks (HetSNets)},
  2012.

\bibitem{Huang:14}
W.~Huang, X.~Jia, and Y.~Zhang, ``{Interference management and traffic
  adaptation of femto base station based on TD-LTE},'' \emph{International
  Journal of Future Generation Communication and Networking}, vol.~7, no.~1,
  pp. 217--224, 2014.

\bibitem{3GPP36828}
3GPP, ``{Further enhancements to LTE time division duplex (TDD) for
  downlink-uplink (DL-UL) interference management and traffic adaptation},'' TR
  36.828 (release 11), 2012.

\bibitem{russellDynamic}
R.~Ford, F.~G{\'{o}}mez-Cuba, M.~Mezzavilla, and S.~Rangan, ``{Dynamic
  time-domain duplexing for self-backhauled millimeter wave cellular
  networks},'' in \emph{IEEE International Conference on Communications (ICC)
  Workshop on Next Generation Backhaul/Fronthaul Networks (BackNets)}, 2015.

\bibitem{juanScheduling}
J.~Garc{\'{i}}a-Rois, F.~G{\'{o}}mez-Cuba, M.~R. Akdeniz, F.~J.
  Gonz{\'{a}}lez-Casta{\~{n}}o, J.~C. Burguillo-Rial, S.~Rangan, and
  B.~Lorenzo, ``{On the analysis of scheduling in dynamic duplex multi-hop
  mmWave cellular systems},'' \emph{IEEE Transactions on Wireless
  Communications}, vol.~14, no.~11, pp. 6028 -- 6042, 2015.

\bibitem{6615900}
Y.~Shi, J.~Liu, C.~Jiang, C.~Gao, and Y.~T. Hou, ``{A DoF-based link layer
  model for multi-hop MIMO networks},'' \emph{IEEE Transactions on Mobile
  Computing}, vol.~13, no.~7, pp. 1395--1408, 2014.

\bibitem{Tuto}
N.~B. Shroff, R.~Srikant, and X.~Lin, ``{A tutorial on cross-layer optimization
  in wireless networks},'' \emph{IEEE Journal on Selected Areas in
  Communications}, vol.~24, no.~8, pp. 1452--1463, 2006.

\bibitem{Rappaport2015}
T.~S. Rappaport, G.~R. Maccartney, M.~K. Samimi, and S.~Sun, ``{Wideband
  millimeter-wave propagation measurements and channel models for future
  wireless communication system design},'' vol.~63, no.~9, pp. 3029 -- 3056,
  2015.

\bibitem{Akdeniz2013}
M.~R. Akdeniz, Y.~Liu, S.~Rangan, and E.~Erkip, ``{Millimeter wave picocellular
  system evaluation for urban deployments},'' in \emph{IEEE Global
  Telecommunications Conference (GLOBECOM)}, 2013.

\bibitem{Akdeniz2014}
M.~R. Akdeniz, Y.~Liu, M.~K. Samimi, S.~Sun, S.~Rangan, T.~S. Rappaport, and
  E.~Erkip, ``{Millimeter wave channel modeling and cellular capacity
  evaluation},'' \emph{IEEE Journal on Selected Areas in Communications},
  vol.~32, no.~6, pp. 1164--1179, 2014.

\bibitem{Barati2015}
C.~N. Barati, S.~A. Hosseini, S.~Rangan, P.~Liu, T.~Korakis, S.~S. Panwar, and
  T.~S. Rappaport, ``{Directional cell discovery in millimeter wave cellular
  networks},'' \emph{IEEE Transactions on Wireless Communications}, vol.~14,
  no.~12, pp. 6664--6678, 2015.

\bibitem{Hur2013}
S.~Hur, T.~Kim, D.~J. Love, J.~V. Krogmeier, T.~A. Thomas, and A.~Ghosh,
  ``{Millimeter wave beamforming for wireless backhaul and access in small cell
  networks},'' \emph{IEEE Transactions on Communications}, vol.~61, no.~10, pp.
  4391--4403, 2013.

\bibitem{Rappaport2014mimo}
S.~Sun, T.~Rappaport, R.~Heath, A.~Nix, and S.~Rangan, ``{MIMO for
  millimeter-wave wireless communications: Beamforming, spatial multiplexing,
  or both?}'' \emph{IEEE Communications Magazine}, vol.~52, no.~12, pp.
  110--121, 2014.

\bibitem{Samsung2014}
W.~Roh, J.-Y. Seol, J.~Park, B.~Lee, J.~Lee, Y.~Kim, J.~Cho, K.~Cheun, and
  F.~Aryanfar, ``{Millimeter-wave beamforming as an enabling technology for 5G
  cellular communications: Theoretical feasibility and prototype results},''
  \emph{IEEE Communications Magazine}, vol.~52, no.~2, pp. 106--113, 2014.

\bibitem{Kutty2015}
S.~Kutty and D.~Sen, ``{Beamforming for millimeter wave communications: An
  inclusive survey},'' \emph{IEEE Communications Surveys {\&} Tutorials},
  vol.~18, no.~2, pp. 949--973, 2016.

\bibitem{Gesbert2007a}
D.~Gesbert, M.~Kountouris, R.~W. Heath, C.~B. Chae, and T.~S{\"{a}}lzer,
  ``{Shifting the MIMO paradigm},'' \emph{IEEE Signal Processing Magazine},
  vol.~24, no.~5, pp. 36--46, 2007.

\bibitem{hoydis2011massive}
J.~Hoydis, S.~{Ten Brink}, and M.~Debbah, ``{Massive MIMO: How many antennas do
  we need?}'' in \emph{IEEE 49th Annual Allerton Conference on Communication,
  Control, and Computing}, 2011, pp. 545--550.

\bibitem{Bjornson2016}
E.~Bj{\"{o}}rnson, E.~G. Larsson, and T.~L. Marzetta, ``{Massive MIMO: ten
  myths and one critical question},'' \emph{IEEE Communications Magazine},
  vol.~54, no.~2, pp. 114--123, 2016.

\bibitem{Orhan2015}
O.~Orhan, E.~Erkip, and S.~Rangan, ``{Low power analog-to-digital conversion in
  millimeter wave systems: Impact of resolution and bandwidth on
  performance},'' in \emph{IEEE Information Theory and Applications Workshop,
  ITA}, 2015, pp. 191--198.

\bibitem{Mo2016}
J.~Mo, A.~Alkhateeb, S.~Abu-Surra, and R.~W. Heath, ``{Hybrid architectures
  with few-bit ADC receivers: Achievable rates and energy-rate tradeoffs},''
  \emph{IEEE Transactions on Wireless Communications}, vol.~16, no.~4, pp. 2274
  -- 2287, 2017.

\bibitem{Abbas2016}
W.~bin Abbas and M.~Zorzi, ``{Towards an appropriate receiver beamforming
  scheme for millimeter wave communication: A power consumption based
  comparison},'' in \emph{22nd European Wireless Conference}, 2016.

\bibitem{Tass}
L.~Tassiulas, ``{Linear complexity algorithms for maximum throughput in radio
  networks and input queued switches},'' in \emph{IEEE INFOCOM}, vol.~2, 1998,
  pp. 533--539.

\bibitem{Kelly1997}
F.~P. Kelly, ``{Charging and rate control for elastic traffic},''
  \emph{European Transactions on Telecommunications}, vol.~8, pp. 33--37, 1997.

\bibitem{Kelly1998}
F.~P. Kelly, A.~K. Maulloo, and D.~K.~H. Tan, ``{Rate control in communication
  networks: Shadow prices, proportional fairness and stability},''
  \emph{Journal of the Operational Research Society}, vol.~49, pp. 237--252,
  1998.

\bibitem{Eryilmaz2007}
A.~Eryilmaz, A.~Ozdaglar, and E.~Modiano, ``{Polynomial complexity algorithms
  for full utilization of multi-hop wireless networks},'' in \emph{IEEE
  INFOCOM}, 2007, pp. 499--507.

\bibitem{ModianoPower}
H.-W. Lee, E.~Modiano, and L.~B. Le, ``{Distributed throughput maximization in
  wireless networks via random power allocation},'' \emph{IEEE Transactions on
  Mobile Computing}, vol.~11, no.~4, 2012.

\bibitem{ZhouDelay2012}
A.~Zhou, M.~Liu, Z.~Li, and E.~Dutkiewicz, ``{Cross-layer design for
  proportional delay differentiation and network utility maximization in
  multi-hop wireless networks},'' \emph{IEEE Transactions on Wireless
  Communications}, vol.~11, no.~4, pp. 1446--1455, 2012.

\bibitem{Mudumbai2009}
R.~Mudumbai, S.~K. Singh, and U.~Madhow, ``{Medium access control for 60 GHz
  outdoor mesh networks with highly directional links},'' in \emph{IEEE
  INFOCOM}, 2009.

\bibitem{yedidia2011message}
J.~S. Yedidia, ``{Message-passing algorithms for inference and optimization},''
  \emph{Journal of Statistical Physics}, vol. 145, no.~4, pp. 860--890, 2011.

\bibitem{moallemi2010convergence}
C.~C. Moallemi and B.~{Van Roy}, ``{Convergence of min-sum Message-Passing for
  convex optimization},'' \emph{IEEE Transactions on Information Theory},
  vol.~56, no.~4, pp. 2041--2050, 2010.

\bibitem{Tweedie83Markov}
\BIBentryALTinterwordspacing
R.~L. Tweedie, ``{The Existence of Moments for Stationary Markov Chains},''
  \emph{Journal of Applied Probability}, vol.~20, no.~1, pp. 191--196, 1983.
  [Online]. Available: \url{http://www.jstor.org/stable/3213735}
\BIBentrySTDinterwordspacing

\bibitem{Eryilmaz2010Implementation}
A.~Eryilmaz, A.~Ozdaglar, D.~Shah, and E.~Modiano, ``{Distributed Cross-Layer
  Algorithms for the Optimal Control of Multihop Wireless Networks},''
  \emph{IEEE/ACM Transactions on Networking}, vol.~18, no.~2, pp. 638--651, apr
  2010.

\end{thebibliography}

\end{document}